\documentclass[a4paper,UKenglish,cleveref,autoref,thm-restate]{lipics-v2021}

\pdfoutput=1 \hideLIPIcs

\title{Counting on Nowhere Dense Classes}

\author{Steffen {van Bergerem}}
       {Humboldt-Universität zu Berlin, Germany}
       {steffen.van.bergerem@hu-berlin.de}
       {https://orcid.org/0000-0002-5212-8992}{}
\author{Nicole Schweikardt}
       {Humboldt-Universität zu Berlin, Germany}
       {schweikn@hu-berlin.de}
       {https://orcid.org/0000-0001-5705-1675}{}

\authorrunning{S.\ van Bergerem and N.\ Schweikardt}

\Copyright{Steffen van Bergerem and Nicole Schweikardt}

\ccsdesc[500]{Theory of computation~Complexity theory and logic}
\ccsdesc[500]{Theory of computation~Finite Model Theory}
\ccsdesc[500]{Theory of computation~Fixed parameter tractability}

\keywords{First-order logic with counting, algorithmic metatheorems,
nowhere dense classes, fixed-parameter almost-linear time,
query answering, constant-delay enumeration}

\category{} 

\relatedversion{}

\useofai{The research presented in this paper was conducted
by the authors without the use of AI tools.
After completing the initial draft, AI tools were used for proofreading.}

\funding{This work was funded by the Deutsche Forschungsgemeinschaft
(DFG, German Research Foundation) -- project number 541000908
(gefördert durch die Deutsche Forschungsgemeinschaft (DFG) -- Projektnummer 541000908).}

\nolinenumbers

\usepackage{stmaryrd}
\usepackage{xspace}
\usepackage{aligned-overset}
\usepackage[capitalise,noabbrev]{cleveref}
\usepackage{csquotes}

\newcommand{\N}{\ensuremath{\mathbb{N}}}
\newcommand{\Npos}{\ensuremath{\mathbb{N}_{\geq 1}}}
\newcommand{\Z}{\ensuremath{\mathbb{Z}}}
\newcommand{\Q}{\ensuremath{\mathbb{Q}}}
\newcommand{\Qpos}{\ensuremath{\mathbb{Q}_{> 0}}}

\newcommand{\A}{\ensuremath{\mathcal{A}}}
\newcommand{\B}{\ensuremath{\mathcal{B}}}
\newcommand{\G}{\ensuremath{\mathcal{G}}}

\newcommand{\CM}{\ensuremath{\mathcal{M}}}

\DeclareMathOperator{\components}{comp}

\newcommand{\relabel}[1]{\rho_{#1}}
\DeclareMathOperator{\std}{std}

\newcommand{\NC}{\mathcal{K}}
\newcommand{\centres}{X}
\newcommand{\centre}{c}
\newcommand{\cluster}{K}
\newcommand{\tCentres}{\tc}
\DeclareMathOperator{\vertexToCentre}{\chi}
\DeclareMathOperator{\centreToCluster}{\kappa}
\newcommand{\vertexToCluster}{\centreToCluster\!\vertexToCentre}
\newcommand{\NCtuple}[1]{\NC#1 = (\centres#1, \centreToCluster#1, \vertexToCentre#1)}
\DeclareMathOperator{\overlap}{overlap}
\DeclareMathOperator{\WReach}{WReach}

\newcommand{\removal}{{\circledcirc}}
\newcommand{\removaldist}[1]{\removal_{#1}}
\newcommand{\removaldistr}{\removaldist{r}}
\newcommand{\removalsig}[2]{#1^{(#2)}}
\newcommand{\removalsigr}[1]{\removalsig{#1}{r}}
\newcommand{\removalsigrs}[1]{\removalsig{#1}{r,s}}

\renewcommand{\phi}{\varphi}
\renewcommand{\epsilon}{\varepsilon}

\newcommand{\logic}[1]{\ensuremath{\mathsf{#1}}}
\newcommand{\FO}{\logic{FO}}
\newcommand{\FOplus}{\FO^+}
\newcommand{\FOplusParams}[3]{\FOplus[#1,#2,#3]}
\newcommand{\FOplusSigma}{\FOplus[\sigma]}
\newcommand{\FOplusSigmaParams}[2]{\FOplusParams{\sigma}{#1}{#2}}

\newcommand{\FOplusSigmaPQ}{\FOplusSigmaParams{p}{q}}
\newcommand{\FOC}{\logic{FOC}}
\newcommand{\cgFOC}{\logic{cgFOC}}
\newcommand{\radius}{r^*}

\newcommand{\vars}{\logic{vars}}
\newcommand{\sem}[1]{\llbracket #1 \rrbracket} \newcommand{\Land}{\ensuremath{\bigwedge}}
\newcommand{\Lor}{\ensuremath{\bigvee}}
\newcommand{\FOCCount}[2]{\ensuremath{\# {#1}.{#2}}}
\newcommand{\I}{\ensuremath{\mathcal{I}}} 

\DeclareMathOperator{\ar}{ar}
\DeclareMathOperator*{\free}{free}
\DeclareMathOperator{\polylog}{polylog}
\DeclareMathOperator{\dom}{dom}

\newcommand{\complexityclass}[1]{\ensuremath{\mathsf{#1}}}
\newcommand{\FPT}{\ensuremath{\complexityclass{FPT}}}
\newcommand{\AWstar}{\ensuremath{\complexityclass{AW}[*]}}

\newcommand{\bigO}{\ensuremath{\mathcal{O}}}

\newcommand{\bigmid}{:}

\newcommand{\set}[1]{\ensuremath{\{#1\}}} \newcommand{\setc}[2]{\ensuremath{\set{#1:#2}}} \newcommand{\bigset}[1]{\ensuremath{\bigl\{ #1 \bigr\}}}
\newcommand{\bigsetc}[2]{\bigset{#1 \bigmid #2}}

\newcommand{\abs}[1]{\left\lvert#1\right\rvert}
\newcommand{\bigabs}[1]{\bigl\lvert#1\bigr\rvert}
\newcommand{\smallabs}[1]{\lvert#1\rvert}
\newcommand{\norm}[1]{\left\lVert#1\right\rVert}
\newcommand{\deff}{\coloneqq}

\DeclareMathOperator{\dist}{dist}

\newcommand{\ie}{\mbox{i.e.}\xspace}
\newcommand{\eg}{\mbox{e.g.}\xspace}

\newcommand{\C}{\ensuremath{\mathcal{C}}}
\newcommand{\PP}{\ensuremath{\mathbb{P}}}
\newcommand{\Pred}{\ensuremath{\mathsf{P}}}

\newcommand{\tuple}[1]{\ensuremath{\bar{#1}}}
\newcommand{\ta}{\tuple{a}}

\newcommand{\tc}{\tuple{c}}

\newcommand{\tu}{\tuple{u}}
\newcommand{\tv}{\tuple{v}}
\newcommand{\tw}{\tuple{w}}
\newcommand{\tx}{\tuple{x}}
\newcommand{\ty}{\tuple{y}}
\newcommand{\tz}{\tuple{z}}

\newcommand{\neighb}[3]{\ensuremath{N_{#1}^{#2}(#3)}} \newcommand{\neighbr}[2]{\neighb{r}{#1}{#2}} \newcommand{\neighbA}[2]{\neighb{#1}{\A}{#2}} \newcommand{\Neighb}[3]{\ensuremath{\mathcal{N}_{#1}^{#2}(#3)}} \newcommand{\Neighbr}[2]{\Neighb{r}{#1}{#2}}  \newcommand{\nrA}[1]{\ensuremath{\neighb{r}{\A}{#1}}}
\newcommand{\NrA}[1]{\ensuremath{\Neighb{r}{\A}{#1}}}

\newcommand{\fNC}{f_{\textup{NC}}}
\newcommand{\fFOTE}{f_{\textup{TE}}}
\newcommand{\fSize}{f_{\textup{size}}}
\newcommand{\fGameTree}{f_{\textup{GT}}}

\newcommand{\fLeaf}{f_{\textup{leaf}}}
\newcommand{\fNode}{f_{\textup{node}}}
\newcommand{\fCount}{f_{\textup{count}}}
\newcommand{\fCountQA}{f_{\textup{countQA}}}
\newcommand{\fPre}{f_{\textup{pre}}}
\newcommand{\fRemoval}{f_{\textup{removal}}}
\newcommand{\fStore}{f_{\textup{store}}}
\newcommand{\fSub}{f_{\textup{sub}}}
\newcommand{\fcgFOC}{f_{\cgFOC}}
 
\begin{document}

\maketitle

\begin{abstract}
  We establish an algorithmic framework
  for query answering and enumeration in first-order logic with counting
  on effectively nowhere dense classes of relational structures.
  Our main technical contribution is an algorithm that,
  after almost-linear-time preprocessing on a structure \(\A\) from such a class
  and a first-order formula \(\phi(\tx, \ty)\),
  can answer the following queries in constant time:
  given a tuple \(\tv \in A^{\abs{\tx}}\),
  output the number of tuples \(\tw \in A^{\abs{\ty}}\)
  satisfying \(\A \models \phi(\tv,\tw)\).
  In database terms, this is constant-time evaluation of grouped counting queries.
  Building on this, we achieve constant-time query answering
  and constant-delay enumeration for the recently introduced
  clique-guarded first-order logic with counting~(\(\cgFOC\))
  [van Bergerem, Lange, and Schweikardt, MFCS~2026]
  on effectively nowhere dense classes.
  These results unify and generalise previous work on
  nowhere dense classes,
  namely first-order testing and enumeration
  [Schweikardt, Segoufin, and Vigny, JACM~2022]
  and the evaluation of the counting logic~\(\FOC_1\)
  [Grohe and Schweikardt, PODS~2018].
  Moreover, under the assumption \(\FPT \neq \AWstar\),
  our results are essentially optimal:
  they do not extend to subgraph-closed graph classes
  beyond the nowhere dense ones,
  slight generalisations of the clique-guarding mechanism
  are intractable even on very simple sparse classes,
  and very restricted fragments of \(\cgFOC\)
  are already intractable on very simple dense classes.
\end{abstract}

\section{Introduction}
\label{sec:intro}

Model checking, enumeration, and counting are fundamental tasks in database theory.
In this setting, counting typically goes together with \emph{grouping}.
A \emph{grouped counting query} based on a query \(\phi(\tx, \ty)\) returns,
for every tuple \(\tv\) of values for the \emph{group-by} variables \(\tx\),
the number of tuples \(\tw\) such that \(\phi(\tv, \tw)\) holds
(in SQL: \;\texttt{SELECT \(\tx\), COUNT(*) WHERE \(\phi\) GROUP BY \(\tx\)}).
For first-order logic (\(\FO\)), however,
already model checking is intractable:
the parameterised problem, where the query is the parameter,
is complete for the complexity class \(\AWstar\),
cf.~\cite{FlumGrohe_2006_Parameterized}.
A central goal of the research on \emph{algorithmic metatheorems} is therefore
to identify structural properties of the input data, such as sparsity,
that render query evaluation tractable.

In a seminal result, Grohe, Kreutzer, and
Siebertz~\cite{GroheKreutzerSiebertz_2017_NowhereDense}
gave an almost-linear-time algorithm for the model-checking problem for \(\FO\)
on nowhere dense graph classes.
That is, for every nowhere dense graph class \(\C\),
they presented an algorithm that checks whether a given \(\FO\) sentence \(\phi\)
holds in a given graph \(G \in \C\) in time
\(\bigO_{\C, \phi, \epsilon}(n^{1+\epsilon})\) for every \(\epsilon > 0\),
where \(n\) denotes the number of vertices in \(G\).
The subscripts indicate that the constants in the \(\bigO\)-notation
may depend on \(\C\), \(\phi\), and \(\epsilon\).
This result generalises prior results for classes of
bounded degree~\cite{Seese_1996_LinearTimeBoundedDegree}
and of bounded expansion and locally bounded
expansion~\cite{DvorakKralThomas_2013_LocallyBoundedExpansion}
to \emph{nowhere dense} classes,
a robust notion of sparsity introduced by Nešetřil
and Ossona de Mendez~\cite{NesetrilOssonaDeMendez_2011_NowhereDense}.
Furthermore, Dawar and Kreutzer~\cite{DawarKreutzer_2009_SomewhereDense}
showed that for graph classes closed under taking subgraphs,
the result of \cite{GroheKreutzerSiebertz_2017_NowhereDense} is optimal.
That is, for every such class that is \emph{not} nowhere dense,
the \(\FO\) model-checking problem is as hard as the model-checking problem
on the class of all graphs, which is \(\AWstar\)-hard.

Schweikardt, Segoufin, and
Vigny~\cite{SchweikardtSegoufinVigny_2022_Enumeration}
generalised the model-checking result
of~\cite{GroheKreutzerSiebertz_2017_NowhereDense} in two directions:
from model checking to constant-time testing and constant-delay enumeration,
and from graph classes to classes of relational structures.
For nowhere dense graph classes \(\C\),
they showed that for every \(\FO\) formula \(\phi\)
and for every relational structure \(\A\)
with \(G_\A \in \C\),
where \(G_\A\) denotes the Gaifman graph of \(\A\),
after almost-linear-time preprocessing,
one can test in constant time whether a given tuple of elements satisfies
\(\phi\) in \(\A\),
and one can enumerate all tuples satisfying \(\phi\) in \(\A\)
with constant delay, without duplicates.
Generalising the result of~\cite{GroheKreutzerSiebertz_2017_NowhereDense}
in another direction,
Grohe and Schweikardt~\cite{GroheSchweikardt_2018_FOC1} gave
an almost-linear-time algorithm for effectively nowhere dense classes
that counts the number of tuples satisfying a given \(\FO\) formula
on a given relational structure.

Our main technical contribution generalises the \(\FO\) counting result
of~\cite{GroheSchweikardt_2018_FOC1}
and the testing result of~\cite{SchweikardtSegoufinVigny_2022_Enumeration}.
Furthermore, it generalises
\cite[Theorem~31 in the full version]{vanBergeremLangeSchweikardt_2026_cgFOC}
from classes of effectively locally bounded expansion
to effectively nowhere dense classes.
We present an algorithm that runs an almost-linear-time preprocessing step
on a given relational structure \(\A\)
and a given \(\FO\) formula \(\phi(\tx, \ty)\).
After the preprocessing,
whenever given a tuple \(\tv\) of values for the variables \(\tx\),
it takes constant time to count the number of tuples \(\tw\)
of values for the variables \(\ty\) such that \(\A \models \phi(\tv, \tw)\).
That is, it evaluates grouped counting queries in constant time per group.
Formally, we prove the following result,
where \(\sigma(\phi)\) denotes the set of relation symbols occurring in \(\phi\).

\begin{theorem}[Counting]
  \label{thm:fo-counting}
  Let \(\C\) be an effectively nowhere dense graph class.
  There is a computable function \(\fCountQA\) and an algorithm that does the following.
  Given an \(\FO\) formula \(\phi(\tx, \ty)\)
  where \(\tx\) and \(\ty\) are disjoint tuples of pairwise distinct variables,
  a \(\sigma\)-structure \(\A\) for some
  \(\sigma \supseteq \sigma(\phi)\) with \(G_\A \in \C\),
  and given an \(\epsilon \in \Qpos\),
  after preprocessing in time
  \(\fCountQA(\phi, \sigma, \epsilon) \cdot \abs{A}^{1+\epsilon}\),
  the algorithm can answer the following queries in time
  \(\fCountQA(\phi, \sigma, \epsilon)\):
  given a tuple \(\tv \in A^{\abs{\tx}}\),
  output \(\abs{\setc{\tw \in A^{\abs{\ty}}}{\A \models \phi(\tv, \tw)}}\).
\end{theorem}

The \(\FO\) counting result of~\cite{GroheSchweikardt_2018_FOC1}
corresponds to the case \(\abs{\tx} = 0\),
the testing result of~\cite{SchweikardtSegoufinVigny_2022_Enumeration}
corresponds to the case \(\abs{\ty} = 0\),
and \cref{thm:fo-counting} interpolates between them.
Moreover, by combining \cref{thm:fo-counting}
with the constant-delay enumeration algorithm
of~\cite{SchweikardtSegoufinVigny_2022_Enumeration},
the whole table of a grouped counting query
(one row per group with a non-zero count)
can be enumerated with constant delay after almost-linear-time preprocessing.
Indeed, the group keys are \(\FO\)-definable
(by the formula \(\exists \ty\, \phi(\tx, \ty)\)),
and \cref{thm:fo-counting} supplies the count for each key in constant time.
A grouped counting query is thus answered at the cost of its output,
whereas enumerating the up to \(\abs{A}^{\abs{\tx}+\abs{\ty}}\) answers of
\(\phi\) and aggregating them afterwards
can be more expensive by a factor of \(\abs{A}^{\abs{\ty}}\).
Hence, on effectively nowhere dense classes,
aggregation comes at no extra cost
compared to the constant-delay enumeration of answers.

The main difficulty in generalising the counting result
of~\cite{GroheSchweikardt_2018_FOC1} is that the entries of
a query tuple \(\tv\) can be scattered across the structure
so that even the count of a single group involves computations
in various regions of \(\A\) that all depend on each other.
Moreover, there are up to \(\abs{A}^{\abs{\tx}}\) possible queries,
all of which must be answerable in constant time
after a single almost-linear-time preprocessing step.

Our proof combines four ingredients.
(1) We organise the \emph{splitter-game} technique
of~\cite{GroheKreutzerSiebertz_2017_NowhereDense} into \emph{game trees},
that is,
rooted trees where root-to-leaf paths correspond to plays of the splitter game
and where every inner node carries a \emph{sparse neighbourhood cover}
of the current graph.
On effectively nowhere dense classes,
these trees have almost-linear size,
and their height is bounded in terms of the input formula
\(\phi\) (\cref{lem:game-tree}).
(2) We use the very recent \emph{rank-preserving Gaifman normal form}
of~\cite{GroheSchweikardt_2026_Locality} to reduce to local formulae.
(3) We formalise an inclusion--exclusion argument to split the counting problem
into independent subproblems on connected regions of the structure
(\cref{lem:reduction-connected}).
(4) A \emph{Removal Lemma} lets us delete vertices from the structure
while preserving all counts by adding relations
that encode the deleted information.
This allows us to make progress in the splitter game and move along the game tree,
where the leaves correspond to structures of small size
on which we solve the counting problem by brute force.
We believe that the formalisation of algorithmic techniques
on nowhere dense classes as game trees (\cref{lem:game-tree})
and the explicit statement of the inclusion--exclusion argument
(\cref{lem:reduction-connected}) are of independent interest.

Our second main result, building on \cref{thm:fo-counting},
generalises the enumeration result
of~\cite{SchweikardtSegoufinVigny_2022_Enumeration}
to formulae of the logic
\emph{clique-guarded first-order logic with counting} (\(\cgFOC\)),
and it generalises \cref{thm:fo-counting} to \emph{expressions},
that is, formulae or \emph{counting terms}, of \(\cgFOC\).
The logic \(\cgFOC\) has been introduced recently by van Bergerem, Lange,
and Schweikardt~\cite{vanBergeremLangeSchweikardt_2026_cgFOC}.
For an expression \(\xi(\tx)\), a structure \(\A\),
and a tuple \(\tv \in A^{\abs{\tx}}\),
we let \(\sem{\xi(\tv)}^\A\) denote the result of evaluating \(\xi\) on \(\A\)
with the variables in \(\tx\) interpreted by \(\tv\)
(see \cref{sec:cgfoc} for a formal definition).
Furthermore, we let \(\sigma(\xi)\) denote the set of relation symbols
occurring in \(\xi\).

\begin{restatable}[Query Answering and Enumeration for cgFOC]{theorem}{answeringAndEnumeration}
  \label{thm:answering-enumeration}
  Let \(\C\) be an effectively nowhere dense graph class.
  There is a computable function \(\fcgFOC\) and an algorithm that does the following.
  Given a \(\cgFOC\) expression \(\xi(\tx)\),
  a \(\sigma\)-structure \(\A\) for some \(\sigma \supseteq \sigma(\xi)\)
  with \(G_\A \in \C\),
  and given an \(\epsilon \in \Qpos\),
  after preprocessing in time \(\fcgFOC(\xi, \sigma, \epsilon) \cdot \abs{A}^{1+\epsilon}\),
  \begin{alphaenumerate}
    \item\label{item:answering}
      the algorithm can answer the following queries in time \(\fcgFOC(\xi, \sigma, \epsilon)\):
      given a tuple \(\tv \in A^{\abs{\tx}}\), output \(\sem{\xi(\tv)}^\A\), and
    \item\label{item:enumeration}
      if \(\xi\) is a formula, then the algorithm can enumerate
      all tuples \(\tv \in A^{\abs{\tx}}\) such that \(\A \models \xi(\tv)\)
      with \(\fcgFOC(\xi, \sigma, \epsilon)\) delay in lexicographic order, without duplicates.
  \end{alphaenumerate}
\end{restatable}

The logic \(\cgFOC\) is a fragment of the first-order logic with counting \(\FOC\).
The logic \(\FOC\) has been introduced by Kuske and Schweikardt~\cite{KuskeSchweikardt_2017_FOCN},
and it extends first-order logic by \emph{counting terms} and \emph{numerical predicates}.
Counting terms are polynomials of \emph{$\#$-terms}
of the form \(\FOCCount{\ty}{\phi(\tx, \ty)}\),
which correspond to grouped counting queries for a formula~\(\phi\)
with group-by variables \(\tx\).
The numerical predicates can be used to compare counting terms;
\ie, \(\Pred(t_1, \dots, t_m)\) is a formula,
where \(t_1, \dots, t_m\) are counting terms
and \(\Pred \subseteq \Z^m\) is a numerical predicate.
For example, the numerical predicate \(\Pred_= \deff \setc{(i,i)}{i \in \Z}\)
checks that the results of two counting terms are equal.
In database terms, this logic can express queries
that filter groups by conditions on the numbers of answers of subqueries
(in SQL: a \texttt{HAVING} clause comparing the results of aggregate functions).

Kuske and Schweikardt~\cite{KuskeSchweikardt_2017_FOCN}
gave algorithms for constant-time testing
and constant-delay enumeration for \(\FOC\) on classes of bounded degree.
Beyond classes of bounded degree,
Grohe and Schweikardt~\cite{GroheSchweikardt_2018_FOC1}
showed that \(\FOC\) model checking
is already \(\AWstar\)-hard on the very restrictive class of
(unranked) trees of height \(3\).
The authors also defined the fragment \(\FOC_1\) of \(\FOC\)
and showed that \(\FOC_1\) model checking
can be solved in almost-linear time
on effectively nowhere dense classes \cite{GroheSchweikardt_2018_FOC1}.
The logic \(\FOC_1\) restricts \(\FOC\) by requiring that the counting terms
\(t_1, \dots, t_m\) in formulae of the form \(\Pred(t_1, \dots, t_m)\)
have at most one free variable in total.
The logic \(\cgFOC\) of~\cite{vanBergeremLangeSchweikardt_2026_cgFOC}
generalises \(\FOC_1\) by requiring that every pair of free variables
of the counting terms \(t_1, \dots, t_m\) is guarded by an atomic formula.
For example, for \(\FO\) formulae \(\phi\) and \(\psi\),
the formula \(E(x,y) \land \Pred_=(\#(z).\phi(x,z), \#(z).\psi(y,z))\)
is a \(\cgFOC\) formula that is not contained in \(\FOC_1\);
it is satisfied by a pair \((x,y)\) of elements if and only if \(E(x,y)\) holds
and the number of elements \(z\) satisfying \(\phi(x,z)\)
equals the number of elements \(z\) satisfying \(\psi(y,z)\).
The guarding mechanism of \(\cgFOC\) ensures that the free variables
of the counting terms compared by a numerical predicate
form a clique in the Gaifman graph.
In fact, the guards can be chosen to be satisfied \emph{if and only if}
the free variables form a clique,
cf.~\cite[Remark~3]{vanBergeremLangeSchweikardt_2026_cgFOC}.

\Cref{thm:answering-enumeration} generalises the result of
\cite{GroheSchweikardt_2018_FOC1} from \(\FOC_1\) to the stronger logic \(\cgFOC\)
and from model checking and counting
to constant-time query answering and constant-delay enumeration.
Furthermore, it subsumes the testing and enumeration results
of~\cite{SchweikardtSegoufinVigny_2022_Enumeration}.
In addition to introducing the logic \(\cgFOC\),
\cite{vanBergeremLangeSchweikardt_2026_cgFOC}
gave algorithms for constant-time query answering and constant-delay enumeration
for \(\cgFOC\) on classes of effectively locally bounded expansion;
\cref{thm:answering-enumeration} generalises these results
to effectively nowhere dense classes,
and it is optimal in the following sense.
\begin{romanenumerate}
  \item
    The \(\FO\) model-checking problem is \(\AWstar\)-hard on every graph class
    that is closed under taking subgraphs and not nowhere
    dense~\cite{DawarKreutzer_2009_SomewhereDense}.
    Since \(\cgFOC\) extends \(\FO\), the same holds for \(\cgFOC\);
    and since \cref{thm:answering-enumeration} generalises model checking,
    our results cannot be extended to subgraph-closed classes
    beyond the nowhere dense ones.
  \item
    Even a slight generalisation of the clique-guarding mechanism
    leads to intractability on very simple sparse classes:
    \cite{vanBergeremLangeSchweikardt_2026_cgFOC} shows that
    if free variables may be guarded by a path of length 2
    instead of a single atom or a clique,
    \ie, if constructs of the form
    \(E(x,y) \land E(y,z) \land \Pred_=(t_1(x), t_2(z))\) are admitted,
    then model checking (parameterised by the input formula)
    for the resulting logic is \(\AWstar\)-hard
    on a class of coloured trees of height at most \(2\).
  \item
    Even a very restricted fragment of \(\cgFOC\) is intractable
    on very simple dense classes:
    \cite{vanBergeremLangeSchweikardt_2026_cgFOC} shows that
    model checking for formulae using only constructs of the form
    \(E(x,y) \land \Pred_=(t_1(x), t_2(y))\)
    is~\(\AWstar\)-hard
    on a class of coloured graphs of shrub-depth at most \(2\).
\end{romanenumerate}

As a further consequence of \cref{thm:answering-enumeration},
the algorithms of~\cite{vanBergeremLangeSchweikardt_2026_cgFOC}
for probably approximately correct (PAC) learning
of \(\cgFOC\)-definable concepts
carry over from classes of effectively locally bounded expansion
to effectively nowhere dense classes.

Orthogonal to our results,
Dreier, Mock, and Rossmanith~\cite{DreierMockRossmanith_2023_FOC}
gave an almost-linear-time algorithm
on nowhere dense graph classes \(\C\) that,
given a \emph{quantifier-free} \(\FO\) formula \(\phi(\tx, y)\)
and a graph \(G \in \C\),
computes a tuple \(\tv \in \bigl(V(G)\bigr)^{\smallabs{\tx}}\)
that maximises the number of vertices \(w \in V(G)\) with
\(G \models \phi(\tv, w)\).
This can be used for model checking formulae of the form
\(\exists x_1 \dots \exists x_k \bigl(\#(y).\phi(\tx, y) > N\bigr)\).
In contrast to our results,
this requires the formula \(\phi\) to be quantifier-free
and only allows counting over a single variable instead of tuples of variables,
but it allows free variables that do not need to be guarded.
On classes of bounded expansion,
\cite{DreierMockRossmanith_2023_FOC} generalised the result
from quantifier-free \(\FO\) to full \(\FO\).
Moreover, also for classes of bounded expansion,
Mock and Rossmanith~\cite{MockRossmanith_2024_BoundedExpansionFOC}
gave an algorithm for checking formulae of the form
\(\psi = \exists x_1 \dots \exists x_k
\,\Pred\bigl(\#(y).\phi_1(\tx, y), \dots, \#(y).\phi_\ell(\tx, y)\bigr)\)
on \(n\)-vertex graphs in time
\(f(\abs{\psi}) \cdot n^{\ell+1} \cdot \polylog(n)\).

The remainder of this paper is organised as follows.
In \cref{sec:prelims},
we introduce nowhere dense classes via the game characterisation
of~\cite{GroheKreutzerSiebertz_2017_NowhereDense}
and present the rank-preserving Gaifman normal form
of~\cite{GroheSchweikardt_2026_Locality}.
In \cref{sec:counting},
we introduce game trees and discuss their computation on nowhere dense classes
(\cref{lem:game-tree}),
we formalise the inclusion--exclusion argument (\cref{lem:reduction-connected}),
and we state an iterated version of the Removal Lemma (\cref{cor:removal}).
At the end of the section,
we combine these results to prove \cref{thm:fo-counting}.
In \cref{sec:cgfoc}, we formally define the logic \(\cgFOC\)
and prove \cref{thm:answering-enumeration} based on \cref{thm:fo-counting}.
Lastly, in \cref{sec:game-tree-proof,sec:tools-formulae-proofs},
we provide proofs for the tools discussed in \cref{sec:counting},
and we conclude in \cref{sec:conclusion}.
 \section{Preliminaries}
\label{sec:prelims}

We let \(\Z\), \(\N\), \(\Npos\), \(\Q\), and \(\Qpos\) denote the sets of
integers, non-negative integers, positive integers, rationals,
and positive rationals, respectively.
For \(m, n \in \Z\), we let \([m, n] \deff \setc{i \in \Z}{m \leq i \leq n}\)
and \([n] \deff [1, n]\).
For a set \(X\) and \(k \in \N\),
we let \(\binom{X}{k} \deff \setc{Y}{Y \subseteq X, \abs{Y} = k}\).

For a \(k\)-tuple \(\tv = (v_1, \dots, v_k)\),
we write \(\abs{\tv}\) to denote its \emph{length} \(k\),
and we let \(\tilde{v} \deff \set{v_1, \dots, v_k}\).
For \(k \in \N\), a \(k\)-tuple \(\tv = (v_1, \dots, v_k)\),
and for a set \(I \subseteq \Z\),
we let \(\tv_I\) be the tuple obtained from \(\tv\)
by dropping all entries \(v_i\) with \(i \not\in I\).
For \(\ell \in \N\) with \(\ell \geq k\)
and for a function \(\mu \colon [\ell] \to S\) for some set \(S \subseteq \Z\),
we let \(\tv_\mu \deff \tv_I\) for \(I \deff \setc{i \in [\ell]}{\mu(i) = 0}\);
note that indices \(i > k\) are simply ignored here,
since \(\tv\) has no entries at such positions.

\paragraph*{Graphs and Structures}
When speaking of \emph{graphs}, we mean finite undirected graphs
without self-loops.
For a graph \(G\), we write \(V(G)\) and \(E(G)\)
to denote its vertex set and edge set, respectively.
For \(V' \subseteq V(G)\),
we write \(G[V']\) to denote the subgraph of \(G\) induced by \(V'\).
For \(n \in \Npos\),
let \(\G_n\) be the set of all graphs \(G\) with \(V(G) = [n]\).
For a graph \(G\) with \(V(G) \subseteq \Z\),
we define the following \emph{standardisation}.
Let \(n \deff \abs{V(G)}\)
and \(\set{i_1, \dots, i_n} \deff V(G)\) with \(i_1 < i_2 < \dots < i_n\).
We let \(\relabel{G} \colon V(G) \to [n]\) be the bijection with
\(\relabel{G}(i_s) \deff s\) for all \(s \in [n]\).
Furthermore, we let \(\std(G) \deff \relabel{G}(G) \in \G_n\),
that is, \(\std(G)\) is the graph with \(V\bigl(\std(G)\bigr) = [n]\)
and \(\set{s,t} \in E\bigl(\std(G)\bigr)\)
if and only if \(\set{i_s, i_t} \in E(G)\).

For two vertices \(v,w\) of a graph \(G\), we let \(\dist^G(v,w)\),
called the \emph{distance} between \(v\) and \(w\) in \(G\),
be the \emph{length} (that is, the number of edges)
of a shortest path between \(v\) and \(w\) in \(G\).
If there is no such path, then \(\dist^G(v,w) \deff \infty\).
For \(r \in \N\),
we let \(\neighbr{G}{v} \deff \setc{w \in V(G)}{\dist^G(v,w) \leq r}\)
be the \emph{\(r\)-neighbourhood} of \(v\) in \(G\),
we let \(\Neighbr{G}{v} \deff G[\neighbr{G}{v}]\)
be the corresponding induced subgraph,
and we extend this to tuples \(\tv \in \bigl(V(G)\bigr)^k\) for \(k \in \Npos\)
by setting \(\neighbr{G}{\tv} \deff \bigcup_{i \in [k]} \neighbr{G}{v_i}\)
and \(\Neighbr{G}{\tv} \deff G[\neighbr{G}{\tv}]\).
We let \(\components(G) \subseteq 2^{V(G)}\)
be the vertex sets of the connected components of \(G\).

A \emph{(relational) signature} is a finite set of relation symbols.
Every relation symbol \(R\) has an \emph{arity} \(\ar(R) \in \N\).
Let \(\sigma\) be a signature.
A \emph{\(\sigma\)-structure \(\A\)}
consists of a finite non-empty set \(A\), called the \emph{universe} of \(\A\),
and a relation \(R(\A) \subseteq A^{\ar(R)}\) for every \(R \in \sigma\).
A \emph{(relational) structure} is a \(\sigma\)-structure
for some signature \(\sigma\),
and we denote the signature of a structure \(\A\) by \(\sigma(\A)\).
Unless explicitly stated otherwise,
we use \(\A\) and \(\B\) and variations such as \(\A_1\), \(\A'\),
etc.~to denote structures,
and we denote their universes by \(A\), \(B\), \(A_1\), \(A'\), etc.,
respectively.
For a signature \(\sigma' \supseteq \sigma\), a \(\sigma'\)-structure \(\A'\)
is a \emph{\(\sigma'\)-expansion} of \(\A\) if the universe of \(\A'\) is \(A\)
and it holds that \(R(\A') = R(\A)\) for all \(R \in \sigma\).
If \(\A'\) is a \(\sigma'\)-expansion of \(\A\),
then \(\A\) is the \emph{\(\sigma\)-reduct} of \(\A'\).
We let the \emph{representation size} of a \(\sigma\)-structure \(\A\) be
\(\norm{\A} \deff \abs{A}
+ \sum_{R \in \sigma} \max(\ar(R),1) \cdot \abs{R(\A)}\).

The \emph{Gaifman graph} of a \(\sigma\)-structure \(\A\)
is the graph \(G_\A\) with \(V(G_\A) = A\) and where,
for \(v, w \in A\),
we have \(\set{v,w} \in E(G_\A)\) if and only if \(v \neq w\)
and \(v,w \in \tilde{u}\)
for some \(R \in \sigma\) and some tuple \(\tu \in R(\A)\).
For a number \(r \in \N\), elements \(v, w \in A\),
and a tuple \(\tv \in A^k\) for some \(k \in \Npos\),
the notions of \emph{distance} \(\dist^\A(v,w)\)
and \emph{neighbourhood} \(\nrA{v}\) and \(\nrA{\tv}\)
naturally extend from graphs to structures by considering
the Gaifman graph \(G_\A\) of \(\A\).
We write \(\NrA{v}\) and \(\NrA{\tv}\) to denote the induced substructures
\(\A[\nrA{v}]\) and \(\A[\nrA{\tv}]\) of \(\A\), respectively.
For \(r \in \N\) and a tuple \(\tv \in A^k\) for some \(k \in \Npos\),
we let \(G^\A_{\tv, r} \in \G_k\) be the graph with \(V(G^\A_{\tv, r}) = [k]\)
and \(\set{i,j} \in E(G^\A_{\tv, r})\)
if and only if \(i \neq j\) and \(\dist^\A(v_i, v_j) \leq r\).

For our algorithmic results,
we assume a linear order \(<\) on the universe \(A\).
For example, in our enumeration result,
we will output tuples in lexicographic order
based on the linear order on \(A\),
and, for concreteness, one may think of \(A\)
as a subset of \(\N\) with its natural order.
However, we do not assume that this linear order is realised
by some relation symbol \(R \in \sigma\).
In fact, the classes of structures that we consider in this paper
rule out relation symbols that realise linear orders of arbitrary size.

As the model of computation in this paper,
we use a random-access machine
with addition, multiplication,
and a uniform cost measure.

\paragraph*{The Logic \(\FOplus\)}
Next, we define the logic \(\FOplus[p,q]\)
that was introduced in~\cite{GroheSchweikardt_2026_Locality}.
We assume familiarity with first-order logic (\(\FO\))
and notions such as the free variables \(\free(\phi)\) of a formula \(\phi\).
A \emph{sentence} is a formula without free variables.
For details on \(\FO\) and related notions,
we refer to~\cref{sec:cgfoc}.
The logic \(\FOplus\) extends \(\FO\)
by adding atomic formulae \(\dist(x,y) \leq d\),
called \emph{distance atoms},
for variables \(x,y\) and \(d \in \N\),
and we let \(\A \models \dist(v,w) \leq d\)
for a structure \(\A\) and \(v,w \in A\)
if and only if \(\dist^\A(v,w) \leq d\).
Note that \(\FOplus\) is a purely syntactic extension of \(\FO\),
since the distance atoms can also be expressed with \(\FO\) formulae.

For an \(\FOplus\) formula \(\phi\),
a non-empty tuple \(\tx = (x_1, \dots, x_k)\) of pairwise distinct variables
with \(\free(\phi) \subseteq \set{x_1, \dots, x_k}\),
and a number \(r \in \N\),
we say that \emph{\(\phi(\tx)\) is \(r\)-local (around \(\tx\))}
if for every signature \(\sigma\) with \(\sigma(\phi) \subseteq \sigma\),
every \(\sigma\)-structure \(\A\),
and for all \(\tv \in A^k\),
we have \(\A \models \phi(\tv)\)
if and only if \(\NrA{\tv} \models \phi(\tv)\).
Note that \(r\)-locality implies \(r'\)-locality for every \(r' \geq r\).
Furthermore, using the \emph{coincidence lemma},
one can show that
if \(\phi(\tx)\) is \(r\)-local for some non-empty tuple \(\tx\)
with \(\free(\phi) \subseteq \tilde{x}\),
then \(\phi(\tx')\) is \(r\)-local for all non-empty tuples \(\tx'\)
of pairwise distinct variables with \(\free(\phi) \subseteq \tilde{x}'\)
(see, \eg, \cite[Lemma~18]{vanBergeremGroheLenzSchweikardt_2026_FOW}).
We say that a formula \(\phi\) is \(r\)-local
if \(\phi(\tx)\) is \(r\)-local for some non-empty tuple \(\tx\)
with \(\free(\phi) \subseteq \tilde{x}\).
If two formulae \(\phi, \psi\) are \(r\)-local,
then \(\neg \phi\) and \(\phi \land \psi\) are \(r\)-local as well.

Let \(\radius \colon \N^2 \to \N, (p,q) \mapsto (4q)^p\).
For \(p,q \in \N\) with \(p \leq q\) and a signature \(\sigma\),
we define the subset \(\FOplusSigmaPQ\) of \(\FOplus[\sigma]\) as follows.
We let \(\FOplusSigmaParams{0}{q}\) be the set of all formulae \(\phi\)
such that \(\abs{\free(\phi)} \leq q\)
and \(\phi\) is a Boolean combination of atomic \(\FO[\sigma]\) formulae
and of distance atoms of the form \(\dist(x,y) \leq d\)
with \(d \leq \radius(0,q)\).
For \(p > 0\), we let \(\FOplusSigmaPQ\) be the set of all formulae \(\phi\)
with \(\abs{\free(\phi)} \leq q-p\) such that \(\phi\) is a Boolean combination
of \(\FOplusSigmaParams{p-1}{q}\) formulae,
distance atoms of the form \(\dist(x,y) \leq d\) with \(d \leq \radius(p,q)\),
formulae of the form \(\exists y\, \psi\)
for an \(\FOplusSigmaParams{p-1}{q}\) formula \(\psi\),
and formulae of the form \(\exists y\, \bigl(\dist(x,y) \leq d \land \psi\bigr)\)
for an \(\FOplusSigmaParams{p-1}{q}\) formula \(\psi\) and
\(d \leq \radius(p,q) - \radius(p-1,q)\).
Note that \(\FOplusSigma = \bigcup_{q \in \N, p \in [0,q]} \FOplusSigmaPQ\),
and every formula in \(\FOplusSigmaPQ\) has at most \(q-p\) free variables.
Furthermore, every \(\FO[\sigma]\) formula \(\phi(\tx)\)
is contained in \(\FOplusSigmaPQ\),
where \(p\) is the quantifier rank of \(\phi\)
and \(q \deff p + \abs{\tx}\).

\begin{theorem}[Rank-Preserving Normal Form,
  {\cite[Theorem~7.1]{GroheSchweikardt_2026_Locality}}]
  \label{thm:rank-preserving}
  There is an algorithm that, upon input of a signature \(\sigma\),
  numbers \(p,q \in \N\) with \(p \leq q\),
  and a formula \(\phi(x_1, \dots, x_k) \in \FOplusSigmaPQ\)
  with \(k = \abs{\free(\phi)} \geq 1\),
  sets \(r \deff \radius(p, q)\) and computes for each \(G \in \G_k\)
  a number \(m_G \in \N\)
  and, for each \(i \in [m_G]\), an \(\FO[\sigma]\) sentence \(\xi^i_G\)
  and, for each \(i \in [m_G]\) and each \(I \in \components(G)\),
  an \(r\)-local formula \(\psi^i_{G,I}(\tx_I) \in \FOplusSigmaPQ\)
  such that the following holds.
  \begin{bracketenumerate}
    \item For all \(\sigma\)-structures \(\A\) and all \(\tv \in A^k\),
      we have \(\A \models \phi(\tv)\) if and only if
      there exists an \(i \in [m_G]\) for \(G \deff G^\A_{\tv,r}\)
      such that \(\A \models \xi^i_G\) and \(\A \models \psi^i_{G,I}(\tv_I)\)
      for every connected component \(I\) of \(G\).
    \item For all \(\sigma\)-structures \(\A\) and all \(\tv \in A^k\),
      there is at most one \(i \in [m_G]\) for \(G \deff G^\A_{\tv,r}\)
      such that \(\A \models \xi^i_G\) and \(\A \models \psi^i_{G,I}(\tv_I)\)
      for every connected component \(I\) of \(G\).
  \end{bracketenumerate}
\end{theorem}

\noindent
For \(r \in \N\)
and a graph \(G\) with \(V(G) \subseteq \Z\)
and \(k \deff \abs{V(G)} \in \Npos\), let
\[
  \delta_{G,r}(x_1, \dots, x_k) \deff
  \Land_{\set{i, j} \in E}
  \dist(x_i, x_j) \leq r
  ~\land~
  \Land_{\set{i, j} \in \binom{[k]}{2} \setminus E}
  \neg \dist(x_i, x_j) \leq r,
\]
where \(E \deff E\bigl(\std(G)\bigr)\).
Note that \(\delta_{G, r}(x_1, \dots, x_k) \in \FOplusSigmaPQ\)
for every signature \(\sigma\)
and all \(p,q \in \N\) with \(k + p \leq q\) and \(r \leq \radius(p,q)\),
and \(\delta_{G, r}(x_1, \dots, x_k)\) is \(r\)-local.

For a formula \(\phi(\tx)\), a structure \(\A\),
and a tuple \(\tv \in A^k\) for some \(k \leq \abs{\tx}\),
we let
\(\phi(\A, \tv)
\deff \setc{\tw \in A^{\abs{\tx} - k}}{\A \models \phi(\tv\tw)}\).
Furthermore, we let \(\phi(\A) \deff \phi(\A, ())\).

\paragraph*{Nowhere Dense Classes}
For the definition of nowhere dense classes,
we use the game characterisation of such classes established
in~\cite[Section~4]{GroheKreutzerSiebertz_2017_NowhereDense}.
Here, we use the definitions from
\cite[Section~9.1]{GroheSchweikardt_2026_Locality},
which slightly differ from the original definitions
from \cite{GroheKreutzerSiebertz_2017_NowhereDense}.

Let \(G\) be a graph, and let \(\lambda, r \in \N\).
The \emph{\((\lambda, r)\)-splitter game} on \(G\) is played by two players,
\emph{Splitter} and \emph{Connector}.
Let \(G^{(0)} \deff G\).
For \(i \in \Npos\), in round \(i\) of the game,
Connector selects a vertex \(v^{(i)} \in V(G^{(i-1)})\).
Next, Splitter selects a set \(W^{(i)} \subseteq \neighbr{G^{(i-1)}}{v^{(i)}}\)
such that \(\sum_{j=1}^i \abs{W^{(j)}} \leq \lambda\).
If \(\neighbr{G^{(i-1)}}{v^{(i)}} \setminus W^{(i)} = \emptyset\),
then Splitter wins.
Otherwise, the play continues with
\(G^{(i)} \deff \Neighbr{G^{(i-1)}}{v^{(i)}} \setminus W^{(i)}\).
If the play never ends, then Connector wins.

A class \(\C\) of graphs is \emph{nowhere dense} if for every \(r \in \N\)
there is a \(\lambda(r) \in \N\) such that, for all graphs \(G \in \C\),
Splitter has a winning strategy for the
\((\lambda(r), r)\)-splitter game on \(G\).
If \(\lambda(r)\) is computable from \(r\),
then \(\C\) is \emph{effectively nowhere dense}.

In the \emph{generalised \((\lambda, r)\)-splitter game} on a graph \(G\),
in round \(i\),
Connector chooses a subgraph \(H^{(i)}\) of \(G^{(i-1)}\)
and selects a vertex \(v^{(i)} \in V(H^{(i)})\).
Splitter then selects a set \(W^{(i)} \subseteq \neighbr{H^{(i)}}{v^{(i)}}\)
such that \(\sum_{j=1}^i \abs{W^{(j)}} \leq \lambda\).
Analogously to the \((\lambda, r)\)-splitter game,
Splitter wins if \(\neighbr{H^{(i)}}{v^{(i)}} \setminus W^{(i)} = \emptyset\),
and otherwise, the play continues with
\(G^{(i)} \deff \Neighbr{H^{(i)}}{v^{(i)}} \setminus W^{(i)}\).
The following lemma states that Splitter can also win
the generalised splitter game on nowhere dense classes,
and the strategy can be computed efficiently.

\begin{lemma}[{\cite[Lemma~9.2 and Remark~9.3]{GroheSchweikardt_2026_Locality}}]
  \label{lem:winning-strategy}
  For every (effectively) nowhere dense graph class \(\C\),
  there are (computable) functions \(\lambda, t \colon \N \to \N\)
  such that, for every graph \(G \in \C\),
  Splitter has a winning strategy in the generalised
  \((\lambda(r), r)\)-splitter game to win in at most \(t(r)\) rounds.
  Moreover, the winning strategy can be chosen such that, in round \(i\),
  using data structures computed in the previous rounds,
  Splitter's answer \(W^{(i)}\) can be computed and the data structures
  can be updated in time
  \(\bigO\bigl(\abs{V(H^{(i)})} + \abs{E(H^{(i)})} + ri\bigr)\).
\end{lemma}

A \emph{partial play} in the generalised \((\lambda, r)\)-splitter game
on a graph \(G\) is a sequence of the form
\(\bigl(H^{(1)}, v^{(1)}, W^{(1)}, H^{(2)}, v^{(2)}, W^{(2)}, \dots,
H^{(k)}, v^{(k)}, W^{(k)}\bigr)\)
for some \(k \in \N\), where \(H^{(i)}\) and \(v^{(i)}\)
are the choices made by Connector
and \(W^{(i)}\) is the choice made by Splitter in round \(i\) of a play on \(G\).

For functions \(\lambda, t \colon \N \to \N\),
if Splitter wins the generalised \((\lambda(r), r)\)-splitter game
on a graph \(G\) in at most \(t(r)\) rounds for every \(r \in \N\),
then the same also holds for every subgraph of \(G\).
Hence, in the remainder of this paper, without loss of generality,
we assume all nowhere dense graph classes to be closed under taking subgraphs.

For \(r \in \N\) and a graph \(G\),
we let a \emph{distance-\(r\) neighbourhood cover of \(G\)}
be a tuple \(\NCtuple{}\)
consisting of a set of \emph{centres} \(\centres \subseteq V(G)\),
a mapping \(\centreToCluster \colon \centres \to 2^{V(G)}\)
mapping a centre \(\centre \in \centres\) to a \emph{cluster}
\(\centreToCluster(\centre) \subseteq V(G)\)
with \(\centre \in \centreToCluster(\centre)\)
such that \(G[\centreToCluster(\centre)]\) is connected,
and a mapping \(\vertexToCentre \colon V(G) \to \centres\)
from vertices to centres such that, for every vertex \(v \in V(G)\),
it holds that \(\neighbr{G}{v} \subseteq \centreToCluster(\vertexToCentre(v))\).
To simplify our notation,
we let \(\vertexToCluster \colon V(G) \to 2^{V(G)}\) be the mapping with
\(\vertexToCluster(v) = \centreToCluster(\vertexToCentre(v))\)
for all \(v \in V(G)\).
Intuitively, \(\vertexToCluster\) maps vertices to clusters,
and for every vertex \(v \in V(G)\), the cluster \(\vertexToCluster(v)\)
needs to fully contain the \(r\)-neighbourhood of \(v\).
Our definition is equivalent to the one
in~\cite{GroheKreutzerSiebertz_2017_NowhereDense};
the only difference is that we package a neighbourhood cover
as a tuple \(\NCtuple{}\)
that explicitly assigns every vertex \(v\) to a centre \(\vertexToCentre(v)\)
such that the corresponding cluster contains \(\neighbr{G}{v}\),
whereas a neighbourhood cover in~\cite{GroheKreutzerSiebertz_2017_NowhereDense}
is defined as the mapping \(\vertexToCluster \colon V(G) \to 2^{V(G)}\)
such that \(\neighbr{G}{v} \subseteq \vertexToCluster(v)\)
and \(\vertexToCluster(v)\) is connected in \(G\)
for every \(v \in V(G)\),
without explicitly mentioning the centres.
Our algorithms access the explicit mappings
\(\vertexToCentre\)
from vertices to centres and
\(\centreToCluster\)
from centres to clusters as described below in \cref{thm:neighbourhood-cover}.
The \emph{radius} of \(\NC\) is the minimum \(r' \in \N\) such that
\(\centreToCluster(\centre)
\subseteq \neighb{r'}{G[\centreToCluster(\centre)]}{\centre}\)
for all \(\centre \in \centres\).
Intuitively, if \(\NC\) has radius at most \(r'\),
then every cluster is contained in the \(r'\)-neighbourhood of its centre,
where distances are measured inside the cluster.
In particular, this implies
\begin{equation}
  \label{eq:nc-radius}
  \centreToCluster(\centre) = \neighb{r'}{G[\centreToCluster(\centre)]}{\centre}
  \quad\text{for all}\quad
  \centre \in \centres
\end{equation}
and, equivalently,
\(\vertexToCluster(v)
= \neighb{r'}{G[\vertexToCluster(v)]}{\vertexToCentre(v)}\)
for all \(v \in V(G)\).
The \emph{overlap} of \(\NC\)
(called \emph{maximum degree} in~\cite{GroheKreutzerSiebertz_2017_NowhereDense})
is the maximum number of clusters in \(\NC\)
that contain the same vertex, that is,
\(\overlap(\NC) \deff \max_{v \in V(G)}
\bigabs{\setc{\centre \in \centres}{v \in \centreToCluster(\centre)}}\).
Note that
\begin{equation}
  \label{eq:nc-overlap}
  \sum_{\centre \in \centres} \abs{\centreToCluster(\centre)}
  \leq \abs{V(G)} \cdot \overlap(\NC).
\end{equation}
As shown in~\cite{GroheKreutzerSiebertz_2017_NowhereDense},
for nowhere dense classes,
we can efficiently compute neighbourhood covers
with small radius and small overlap.

\begin{restatable}[{\cite[Theorem~6.2]{GroheKreutzerSiebertz_2017_NowhereDense}}]{theorem}{neighbourhoodCover}
  \label{thm:neighbourhood-cover}
  Let \(\C\) be an (effectively) nowhere dense graph class.
  There is a (computable) function \(\fNC\)
  and an algorithm that does the following.
  Given an \(\epsilon \in \Qpos\), \(r \in \N\), and a graph \(G \in \C\),
  the algorithm computes in time
  \(\fNC(r, \epsilon) \cdot \abs{V(G)}^{1+\epsilon}\)
  a data structure representing
  a distance-\(r\) neighbourhood cover
  \(\NCtuple{}\) of \(G\)
  of radius at most \(2r\)
  and overlap at most \(\fNC(r, \epsilon) \cdot \abs{V(G)}^\epsilon\)
  as follows.

  The data structure contains a list of all centres in \(\centres\).
  Furthermore, given a centre \(\centre \in \centres\),
  we can compute a cluster \(\centreToCluster(\centre)\)
  in time \(\bigO(\abs{\centreToCluster(\centre)})\),
  given a vertex \(v \in V(G)\),
  we can compute \(\vertexToCentre(v)\) in time \(\bigO(1)\),
  and given a vertex \(v \in V(G)\) and a centre \(\centre \in \centres\),
  we can check in time \(\fNC(r, \epsilon)\)
  whether \(v \in \centreToCluster(\centre)\).
\end{restatable}

\Cref{thm:neighbourhood-cover} is stated slightly more generally
than~\cite[Theorem~6.2]{GroheKreutzerSiebertz_2017_NowhereDense},
and it is also more explicit on the involved data structures.
We discuss this in detail in~\cref{sec:details-neighbourhood-cover}.

The following two results have originally been shown for \(\FO\).
However, since \(\FOplus\) only adds syntactic sugar to \(\FO\),
and every \(\FOplus\) formula can be transformed
into an equivalent \(\FO\) formula,
we directly state the results for \(\FOplus\).
The first result generalises the \(\FO\) model-checking result
on effectively nowhere dense classes
from~\cite[Theorems~1.1 and 8.1]{GroheKreutzerSiebertz_2017_NowhereDense}.

\begin{theorem}[Testing and Enumeration,
  {\cite[Lemma~2.2 and Corollaries~2.4 and 2.5]{SchweikardtSegoufinVigny_2022_Enumeration}}]
  \label{thm:fo-testing-enumeration}
  Let \(\C\) be an (effectively) nowhere dense graph class.
  There is a (computable) function \(\fFOTE\)
  and an algorithm that does the following.
  Given an \(\FOplus\) formula \(\phi(\tx)\),
  a \(\sigma\)-structure \(\A\) for some \(\sigma \supseteq \sigma(\phi)\)
  with \(G_\A \in \C\),
  and given an \(\epsilon \in \Qpos\),
  after preprocessing in time
  \(\fFOTE(\phi, \sigma, \epsilon) \cdot \abs{A}^{1+\epsilon}\),
  \begin{alphaenumerate}
    \item\label{item:fo-testing}
      the algorithm can answer the following queries in time
      \(\fFOTE(\phi, \sigma, \epsilon)\):
      given a tuple \(\tv \in A^{\abs{\tx}}\),
      decide whether \(\A \models \phi(\tv)\), and
    \item\label{item:fo-enumeration}
      the algorithm can enumerate all tuples \(\tv \in A^{\abs{\tx}}\)
      such that \(\A \models \phi(\tv)\)
      with \(\fFOTE(\phi, \sigma, \epsilon)\) delay in lexicographic order,
      without duplicates.
  \end{alphaenumerate}
\end{theorem}

In addition, \cite[Corollary~5.6]{GroheSchweikardt_2018_FOC1} implies
the following special case of \cref{thm:fo-counting}
in which the query tuple is empty,
so the algorithm counts all satisfying assignments of \(\phi\).
\begin{theorem}[Counting, {\cite[Corollary~5.6]{GroheSchweikardt_2018_FOC1}}]
  \label{thm:fo-counting-static}
  Let \(\C\) be an effectively nowhere dense graph class.
  There is a computable function \(\fCount\) and an algorithm
  that does the following.
  Given an \(\FOplus\) formula \(\phi(\tx)\),
  a \(\sigma\)-structure \(\A\) for some \(\sigma \supseteq \sigma(\phi)\)
  with \(G_\A \in \C\),
  and given an \(\epsilon \in \Qpos\),
  the algorithm computes \(\abs{\phi(\A)}\)
  in time \(\fCount(\phi, \sigma, \epsilon) \cdot \abs{A}^{1+\epsilon}\).
\end{theorem}

Lastly, we will use the following bounds on the size of a structure
and the number of cliques in nowhere dense classes.

\begin{lemma}[{\cite[Lemma~10]{vanBergeremLangeSchweikardt_2026_cgFOC}}]
  \label{lem:nowhere-dense-cliques}
  Let \(\C\) be an (effectively) nowhere dense graph class.
  There is a (computable) function \(\fSize \colon \N \times \Qpos \to \N\)
  such that
  \begin{enumerate}
    \item
      \label{item:nowhere-dense-cliques}
      for every \(k \in \Npos\), \(\epsilon \in \Qpos\), and \(G \in \C\),
      it holds that \(G\) contains at most
      \(\fSize(k,\epsilon) \cdot \abs{V(G)}^{1+\epsilon}\)
      cliques of size at most \(k\), and
    \item
      \label{item:nowhere-dense-size}
      for every signature \(\sigma\),
      every \(\epsilon \in \Qpos\),
      and every \(\sigma\)-structure \(\A\) with \(G_\A \in \C\),
      it holds that
      \(\norm{\A} \leq \abs{\sigma} \cdot \fSize(k, \epsilon)
      \cdot \abs{A}^{1+\epsilon}\),
      where \(k\) is the maximum arity of a relation symbol in~\(\sigma\).
  \end{enumerate}
\end{lemma}

In particular, by \cref{item:nowhere-dense-size}
of \cref{lem:nowhere-dense-cliques},
an almost-linear running time in terms of \(\abs{A}\)
(as stated in \cref{thm:fo-counting,thm:answering-enumeration})
suffices to read the whole input structure \(\A\).
 \section{First-Order Counting on Nowhere Dense Classes}
\label{sec:counting}

In this section, we prove the following result,
which implies \cref{thm:fo-counting}.

\begin{restatable}{theorem}{foplusCounting}
  \label{thm:foplus-counting}
  Let \(\C\) be an effectively nowhere dense graph class.
  There is a computable function \(f\) and an algorithm that does the following.
  Given \(k, \ell, p, q \in \N\) with \(k+\ell+p \leq q\),
  a signature \(\sigma\),
  an \(\FOplusSigmaPQ\) formula \(\phi(x_1, \dots, x_k, y_1, \dots, y_\ell)\)
  with pairwise distinct variables \(x_1, \dots, x_k, y_1, \dots, y_\ell\),
  a \(\sigma\)-structure \(\A\) with \(G_\A \in \C\),
  and given an \(\epsilon \in \Qpos\),
  after preprocessing in time \(f(\phi, q, \sigma, \epsilon) \cdot \abs{A}^{1+\epsilon}\),
  the algorithm can answer the following queries in time \(f(\phi, q, \sigma, \epsilon)\):
  given a tuple \(\tv \in A^k\), output \(\bigabs{\phi(\A, \tv)}\).
\end{restatable}

As noted in \cref{sec:prelims},
for every signature \(\sigma\)
and every \(\FO[\sigma]\) formula \(\phi(\tx, \ty)\),
we have \(\phi \in \FOplusSigmaPQ\),
where \(p\) is the quantifier rank of \(\phi\)
and \(q \deff p + \abs{\tx} + \abs{\ty}\).
Hence, \cref{thm:fo-counting} follows from \cref{thm:foplus-counting}
by setting
\(\fCountQA(\phi, \sigma, \epsilon) \deff f(\phi, q, \sigma, \epsilon)\),
where \(f\) is the computable function from \cref{thm:foplus-counting}.

As described in \cite[Section~2.5]{GroheSchweikardt_2026_Locality},
there is a computable normal form for \(\FOplus\)
such that, for every formula \(\phi\) in \(\FOplusSigmaPQ\),
we can compute an equivalent formula \(\phi'\) in \(\FOplusSigmaPQ\)
in this normal form.
Furthermore, the number of formulae in \(\FOplusSigmaPQ\) in this normal form
is finitely bounded in terms of \(\sigma\), \(p\), and \(q\),
and the set of all such formulae can be computed
from \(\sigma\), \(p\), and \(q\).
Whenever working with formulae in this section,
we will first translate them into the normal form.
Hence, after initially processing the input formula,
we can express the dependence of running-time bounds
on a formula \(\phi \in \FOplusSigmaPQ\)
also in terms of \(\sigma\), \(p\), and \(q\) instead.

In the following,
we discuss all key ingredients of the proof of \cref{thm:foplus-counting}.
In the preprocessing phase for \cref{thm:foplus-counting},
we build a \enquote{game tree}
based on winning strategies in the splitter game from \cref{lem:winning-strategy}
and neighbourhood covers from \cref{thm:neighbourhood-cover}.
All computations in the remaining preprocessing phase as well as in the query-answering phase
will be based on this game tree, which we now define.

For \(\lambda, r \in \N\) and a graph \(G\),
a \emph{\((\lambda, 2r)\)-game tree for \(G\)}
consists of a rooted tree \(T\);
a graph \(G_u\) and
a distance-\(r\) neighbourhood cover \(\NCtuple{_u}\) of \(G_u\) of radius at most \(2r\)
for every inner (that is, non-leaf) node \(u \in V(T)\);
and a graph \(H_u\), a vertex \(v_u \in V(H_u)\), and a set \(W_u \subseteq V(H_u)\)
for every non-root node \(u \in V(T)\).

Intuitively, every path from the root of \(T\) to a leaf
represents a play in the generalised \((\lambda, 2r)\)-splitter game on \(G\),
and a path from the root to a node at depth \(i\) represents a partial play after \(i\) rounds.
The root of the tree represents the initial partial play \(()\) on \(G\).
For a node \(u\) of the tree at depth \(i\) and the corresponding graph
\(G_u = G^{(i)}\) after \(i\) rounds of the game,
moving to a child of \(u\) in \(T\) corresponds to a single round in the game
where Connector chooses \(v^{(i+1)} \in \centres_u\)
to be a centre of the neighbourhood cover \(\NC_u\),
and \(H^{(i+1)} \deff G^{(i)}[\centreToCluster_u(v^{(i+1)})]\)
is the subgraph of \(G^{(i)}\) induced by the corresponding cluster.
Thus, the tree \(T\) represents all plays where Connector is required to
restrict the current graph to a cluster in the neighbourhood cover
and choose the centre of the cluster as the vertex for the next round,
and Splitter plays according to a fixed strategy.

Formally, we represent the nodes \(u \in V(T)\)
as tuples \(\tCentres\) of centres.
The root of \(T\) is the empty tuple \(()\),
and we have \(G_{()} \deff G\).
Furthermore, \(\NC_{()}\) is a distance-\(r\) neighbourhood cover
of \(G\) of radius at most \(2r\).
Let \(\pi_{()} \deff ()\) be the (initial) partial play
in the generalised \((\lambda, 2r)\)-splitter game on \(G\).

For every inner node \(\tCentres = (\centre_1, \dots, \centre_i)\) of \(T\)
for some \(i \in \N\)
(for \(i = 0\), we have \(\tCentres = ()\))
and every centre \(\centre \in \centres_{\tCentres}\),
we let \(\tCentres\) have a child
\(\tCentres \centre = (\centre_1, \dots, \centre_i, \centre)\)
with \(v_{\tCentres \centre} \deff \centre\),
\(H_{\tCentres \centre} \deff G_{\tCentres}[\centreToCluster_{\tCentres}(\centre)]\),
and \(W_{\tCentres \centre}\) is Splitter's answer according to a fixed strategy
for the generalised \((\lambda, 2r)\)-splitter game on \(G\)
when Connector continues the partial play \(\pi_{\tCentres}\)
by choosing \(H_{\tCentres \centre}\) and \(v_{\tCentres \centre}\).
Let \(\pi_{\tCentres \centre}\) be the resulting partial play after these choices.
If \(W_{\tCentres \centre} = V(H_{\tCentres \centre})\),
then
\(\neighb{2r}{H_{\tCentres \centre}}{v_{\tCentres \centre}}
\setminus W_{\tCentres \centre} = \emptyset\),
so Splitter wins.
In this case, \(\tCentres \centre\) is a leaf of the tree.
Otherwise, we have \(W_{\tCentres \centre} \subsetneq V(H_{\tCentres \centre})\).
Since
\(V(H_{\tCentres \centre})
= \centreToCluster_{\tCentres}(\centre)
\overset{\eqref{eq:nc-radius}}{=}
\neighb{2r}{G_{\tCentres}[\centreToCluster_{\tCentres}(\centre)]}{\centre}
= \neighb{2r}{H_{\tCentres \centre}}{\centre}\),
this implies that
\(\neighb{2r}{H_{\tCentres \centre}}{v_{\tCentres \centre}}
\setminus W_{\tCentres \centre} \neq \emptyset\),
so the play continues, and \(\tCentres \centre\) is not a leaf.
We let
\(G_{\tCentres \centre} \deff
\Neighb{2r}{H_{\tCentres \centre}}{v_{\tCentres \centre}} \setminus W_{\tCentres \centre}
= H_{\tCentres \centre} \setminus W_{\tCentres \centre}\).
Furthermore, we let \(\NC_{\tCentres \centre}\) be a distance-\(r\) neighbourhood cover
of \(G_{\tCentres \centre}\) of radius at most \(2r\).

By combining \cref{lem:winning-strategy} and \cref{thm:neighbourhood-cover},
we prove in \cref{sec:game-tree-proof} that such game trees exist
and can be computed efficiently for nowhere dense classes,
as stated in the following lemma.

\begin{restatable}{lemma}{gameTree}
  \label{lem:game-tree}
  Let \(\C\) be an effectively nowhere dense graph class,
  let \(\lambda, t\) be the computable functions from \cref{lem:winning-strategy},
  and let \(\fNC\) be the computable function from \cref{thm:neighbourhood-cover} for \(\C\).
  There is a computable function \(\fGameTree\) and an algorithm that,
  given a radius \(r \in \N\), a graph \(G \in \C\),
  and \(\epsilon, \epsilon' \in \Qpos\) with \((1+\epsilon')^{t(2r)+1} \leq 1 + \epsilon\),
  computes a \((\lambda(2r), 2r)\)-game tree for \(G\)
  in time \(\fGameTree(r, \epsilon, \epsilon') \cdot \abs{V(G)}^{1+\epsilon}\)
  such that the tree has height at most \(t(2r)\)
  and, for every inner node \(\tCentres\) of the tree,
  the neighbourhood cover \(\NC_{\tCentres}\) has overlap at most
  \(\fNC(r, \epsilon') \cdot \abs{V(G_{\tCentres})}^{\epsilon'}\).
\end{restatable}

When given a \(\sigma\)-structure \(\A\)
and an \(\FOplusSigmaPQ\) formula \(\phi(x_1, \dots, x_k, y_1, \dots, y_\ell)\)
as input for the algorithm for \cref{thm:foplus-counting},
we set \(r \deff \radius(p,q)\), \(r' \deff r \cdot (k+\ell)\),
and we build a \((\lambda(2r'), 2r')\)-game tree for \(G_\A\).
Next, we compute a rank-preserving Gaifman normal form for \(\phi\)
using \cref{thm:rank-preserving},
and we evaluate the resulting \(\FO[\sigma]\) sentences \(\xi^i_G\) on \(\A\)
using \cref{thm:fo-testing-enumeration}.
This turns the evaluation of the \(\FOplusSigmaPQ\) formula \(\phi\)
into an evaluation of \(r\)-local \(\FOplusSigmaPQ\) formulae.
However, although the formulae are local,
an input tuple \(\tv \in A^k\) might be scattered across the structure,
and the evaluation of the local formulae could lead to computations
in various regions of the structure that all depend on each other.
The following result, which we prove in \cref{sec:tools-formulae-proofs}
by the inclusion--exclusion principle,
turns this into independent computations in connected regions of the structure.

\begin{restatable}{lemma}{reductionConnected}
  \label{lem:reduction-connected}
  There is an algorithm that,
  given \(k \in \Npos\) and a graph \(G\)
  with \(V(G) \subseteq \Z\) and \(\abs{V(G)} = k\),
  outputs \(m \in \N\),
  connected graphs \(G_1, \dots, G_m\) with
  \(V_i \deff V(G_i)\) being a union of connected components of \(G\)
  for all \(i \in [m]\),
  and a polynomial \(P(X_1, \dots, X_m) \in \Z[X_1, \dots, X_m]\)
  such that the following holds.
Let \(\sigma\) be a signature,
  let \(p, q, r \in \N\) with \(k+p \leq q\),
  let \(\tx = (x_1, \dots, x_k)\)
  be a tuple of pairwise distinct variables,
  and, for every \(I \in \components(G)\),
  let \(\psi_{G,I}(\tx_{\relabel{G}(I)}) \in \FOplusSigmaPQ\).
  Let
  \begin{align*}
    \phi_G(\tx)
    &\deff \delta_{G,r}(\tx) \land \Land_{I \in \components(G)}
    \psi_{G,I}(\tx_{\relabel{G}(I)}),\\
    \phi_{G,i}(\tx_{\relabel{G}(V_i)})
    &\deff \delta_{G_i, r}(\tx_{\relabel{G}(V_i)}) \land
    \Land_{\substack{I \in \components(G),\\ I \subseteq V_i}}
    \psi_{G,I}(\tx_{\relabel{G}(I)})
    \quad \text{for all } i \in [m],
  \end{align*}
  let \(\A\) be a \(\sigma\)-structure,
  and let \(\tv \in A^\ell\) for some \(\ell \leq k\).
  It holds that
  \[
    \bigabs{\phi_G(\A, \tv)}
    = P\bigl(
      \bigabs{\phi_{G,1}(\A, \tv_{\relabel{G}(V_1)})},
      \dots,
      \bigabs{\phi_{G,m}(\A, \tv_{\relabel{G}(V_m)})}
    \bigr).
  \]
\end{restatable}
Note that, if \(V(G) = [\abs{V(G)}]\),
then \(\relabel{G}(I) = I\) for all \(I \subseteq V(G)\).
Hence, whenever we apply the lemma to such a graph \(G\),
we can omit the application of \(\relabel{G}\).

For the evaluation of the formulae on the connected regions of the structure,
the following result allows us to restrict the computation to a single cluster
in a neighbourhood cover of \(\A\) with a sufficiently large radius.
In fact, a distance-\(r'\) neighbourhood cover suffices for this,
and this is exactly the distance we chose for the neighbourhood covers
in our game tree.
We prove the result in \cref{sec:tools-formulae-proofs}.

\begin{restatable}{lemma}{deltaConnectedLocal}
  \label{lem:delta-connected-local}
  Let \(p \in \N\) and \(k,q \in \Npos\) with \(k + p \leq q\),
  \(r \deff \radius(p,q)\),
  let \(G\) be a connected graph
  with \(V(G) \subseteq \Z\) and \(\abs{V(G)} = k\),
  let \(\sigma\) be a signature,
  let \(\tx \deff (x_1, \dots, x_k)\)
  be a tuple of pairwise distinct variables,
  and let \(\psi(\tx) \in \FOplusSigmaPQ\) be \(r\)-local.
  For \(\phi(\tx) \deff \delta_{G,r}(\tx) \land \psi(\tx)\),
  for every \(\sigma\)-structure \(\A\),
  all \(\tv \in A^k\),
  and every set \(S \subseteq A\) with \(\neighbA{kr}{v_1} \subseteq S\),
  we have \(\A \models \phi(\tv)\)
  if and only if \(\tv \in S^k\) and \(\A[S] \models \phi(\tv)\).
\end{restatable}

All in all, these results turn the problem of counting
the number of satisfying assignments of a formula on the whole structure \(\A\)
into independent counting problems on substructures
\(\A[\centreToCluster_{()}(\centre)]\)
induced by single clusters
of a centre \(\centre \in \centres_{()}\)
of the neighbourhood cover of \(\A\).
By the definition of the game tree,
we know that removing the vertices in \(W_{(\centre)}\)
from \(\A[\centreToCluster_{()}(\centre)]\)
will complete one round of the splitter game,
which is equivalent to moving from the root \(()\)
to one of its children \((\centre)\).
After repeating this process for at most \(t(2r')\) rounds, we reach a leaf,
where the cluster contains at most \(\lambda(2r')\) vertices,
and we simply solve the counting problem on the substructure
induced by the cluster by brute force.
After the preprocessing phase, when given a tuple \(\tv \in A^k\),
we only need to follow constantly many root-to-leaf paths
(where the constant depends on the input formula)
and combine the results bottom-up with precomputed results
that do not depend on \(\tv\).

Thus, the last remaining key tool is a \emph{Removal Lemma},
which can be used to remove vertices from a structure without information loss.
For that, we add new relations that encode the removed information,
and we modify formulae to work with these new relations.

Let \(\sigma\) be a signature, and let \(r \in \N\).
For every relation symbol \(R \in \sigma\) of arity \(k \geq 1\)
and for every set \(I \subsetneq [k]\),
we introduce a fresh \(\abs{I}\)-ary relation symbol \(R_I\).
Moreover, we introduce fresh unary relation symbols \(D_j\)
for all \(j \in [0,r]\).
Let
\[
  \removalsigr{\sigma} \deff
  \sigma
  ~\uplus~
  \setc{R_I}{R \in \sigma, I \subsetneq [\ar(R)]}
  ~\uplus~
  \setc{D_j}{j \in [0,r]}.
\]
For every \(\sigma\)-structure \(\A\) with \(\abs{A} \geq 2\)
and every \(d \in A\),
we let \(\A \removaldistr d\) be the \(\removalsigr{\sigma}\)-structure
with universe \(A \setminus \set{d}\) and relations
\(R(\A \removaldistr d) \deff R(\A) \cap (A \setminus \set{d})^{\ar(R)}\)
for \(R \in \sigma\),
\[R_I(\A \removaldistr d) \deff \bigsetc{\tv_I}
{\tv \in R(\A) \text{ and } I = \setc{i \in [\ar(R)]}{v_i \neq d}}\]
for \(R \in \sigma\) and \(I \subsetneq [\ar(R)]\), and
\[D_j(\A \removaldistr d)
\deff \setc{w \in A \setminus \set{d}}{\dist^{\A}(d,w) \leq j}\]
for \(j \in [0,r]\).
Note that there is a computable function \(\fRemoval\)
such that we can compute \(\A \removaldistr d\)
from \(\A\) and \(d\) in time \(\fRemoval(r, \sigma) \cdot \norm{\A}\).

\begin{lemma}[Removal Lemma, {\cite[Lemma~9.6]{GroheSchweikardt_2026_Locality}}]
  \label{lem:removal}
  Let \(k, p, q \in \N\) with \(k+p \leq q\), let \(r \deff \radius(p, q)\),
  and let \(\sigma\) be a signature.
  For every tuple \(\tz = (z_1, \dots, z_k)\) of pairwise distinct variables,
  every \(\FOplusSigmaPQ\) formula \(\phi(\tz)\),
  and every set \(I \subseteq [k]\),
  there is an \(\FOplusParams{\removalsigr{\sigma}}{p}{q}\) formula \(\phi_I(\tz_I)\)
  such that for all \(\sigma\)-structures \(\A\) with \(\abs{A} \geq 2\),
  all \(d \in A\), and all \(\tv = (v_1, \dots, v_k) \in A^k\) such that
  \(I = \setc{i \in [k]}{v_i \neq d}\), we have
  \(\A \models \phi(\tv) \iff \A \removaldistr d \models \phi_I(\tv_I)\).
  Furthermore, there is an algorithm that computes \(\phi_I(\tz_I)\)
  from \(p\), \(q\), \(\phi(\tz)\), and \(I\).
\end{lemma}

We note that \cite[Lemma~9.6]{GroheSchweikardt_2026_Locality} states the lemma
for the case \(\tz = (x_1, \dots, x_k)\).
However, the more general result stated here is easily obtained by renaming variables.

Analogously to \cite[Corollary~9.8]{GroheSchweikardt_2026_Locality},
which provides an iterated version of the Removal Lemma for formulae with a single free variable,
we do the same for arbitrary formulae.
Let \(\sigma\) be a signature, and let \(r \in \N\).
We let \(\removalsig{\sigma}{r,0} \deff \sigma\), and, for \(s \in \N\),
we let \(\removalsig{\sigma}{r,s+1} \deff \removalsigr{\tau}\)
for \(\tau \deff \removalsigrs{\sigma}\).
Furthermore, for a \(\sigma\)-structure \(\A\) with \(\abs{A} > s\)
and for distinct elements \(d_1, \dots, d_s \in A\),
we let \(\A \removaldistr d_1 \cdots d_s \deff
( \cdots ((\A \removaldistr d_1) \removaldistr d_2) \cdots ) \removaldistr d_s\).
If \(s = 0\), then we let \(\A \removaldistr d_1 \cdots d_s \deff \A\).
The structure \(\A \removaldistr d_1 \cdots d_s\) is a \(\removalsigrs{\sigma}\)-structure
with universe \(A \setminus \set{d_1, \dots, d_s}\)
and Gaifman graph \(G_\A[A \setminus \set{d_1, \dots, d_s}]\).
Moreover, it holds that
\(\sigma = \removalsig{\sigma}{r,0} \subseteq \removalsig{\sigma}{r,1}
\subseteq \dots \subseteq \removalsigrs{\sigma}\).
For the statement of the iterated version,
which we prove in \cref{sec:tools-formulae-proofs},
recall that, for a \(k\)-tuple \(\tv\) and a function
\(\mu \colon [k] \to [0,s]\) for some \(s \in \N\),
the tuple \(\tv_\mu\) is obtained from \(\tv\)
by dropping all entries \(v_i\) with \(\mu(i) \neq 0\).

\begin{restatable}{corollary}{removalIterated}
  \label{cor:removal}
  Let \(k, p, q \in \N\) with \(k+p \leq q\), let \(r \deff \radius(p, q)\),
  let \(\sigma\) be a signature, and let \(s \in \N\).
  For every tuple \(\tz = (z_1, \dots, z_k)\) of pairwise distinct variables,
  every \(\FOplusSigmaPQ\) formula \(\phi(\tz)\),
  and every function \(\mu \colon [k] \to [0,s]\),
  there is an \(\FOplusParams{\removalsigrs{\sigma}}{p}{q}\) formula \(\phi_\mu(\tz_\mu)\)
  such that, for all \(\sigma\)-structures \(\A\) with \(\abs{A} > s\),
  all distinct \(d_1, \dots, d_s \in A\),
  and all \(\tv = (v_1, \dots, v_k) \in A^k\) with
  \(v_i \not\in \set{d_1, \dots, d_s}\) if \(\mu(i) = 0\) and \(v_i = d_{\mu(i)}\) else,
  we have
  \(\A \models \phi(\tv) \iff \A \removaldistr d_1 \cdots d_s \models \phi_\mu(\tv_\mu)\).
  Furthermore, there is an algorithm that computes \(\phi_\mu(\tz_\mu)\)
  from \(p\), \(q\), \(\phi(\tz)\), and \(\mu\).
\end{restatable}

In order to simplify notation, before proving \cref{thm:foplus-counting},
we first prove the following special case
where the formula \(\phi(\tx, \ty)\) is of the form \(\delta_{G, r}(\tx, \ty) \land \psi(\tx, \ty)\)
for a connected graph \(G\)
and an \(r\)-local formula \(\psi(\tx, \ty) \in \FOplusSigmaPQ\).

\begin{lemma}
  \label{lem:foplus-counting-connected}
  Let \(\C\) be an effectively nowhere dense graph class.
  There is a computable function \(f\) and an algorithm that does the following.
  The input of the algorithm consists of numbers
  \(k, \ell, q \in \Npos\) and \(p \in \N\) with \(k+\ell+p \leq q\),
  a signature \(\sigma\),
  a \(\sigma\)-structure \(\A\) with \(G_\A \in \C\),
  an \(\epsilon \in \Qpos\),
  an \(r\)-local \(\FOplusSigmaPQ\) formula \(\psi(\tx, \ty)\),
  and a connected graph \(G \in \G_{k+\ell}\),
  where \(\tx \deff (x_1, \dots, x_k)\)
  and \(\ty \deff (y_1, \dots, y_\ell)\)
  such that \(x_1, \dots, x_k, y_1, \dots, y_\ell\)
  are pairwise distinct variables,
  and \(r \deff \radius(p,q)\).
  Let \(\phi(\tx, \ty) \deff \delta_{G, r}(\tx, \ty) \land \psi(\tx, \ty)\).
  After preprocessing the input in time
  \(f(\phi, q, \sigma, \epsilon) \cdot \abs{A}^{1+\epsilon}\),
  the algorithm can answer the following queries in time
  \(f(\phi, q, \sigma, \epsilon)\):
  given a tuple \(\tv \in A^k\), output \(\bigabs{\phi(\A, \tv)}\).
\end{lemma}
\begin{proof}
  Let \(\lambda, t\) be the computable functions from \cref{lem:winning-strategy},
  let \(\fNC\) be the computable function from \cref{thm:neighbourhood-cover},
  let \(\fSize\) be the computable function from \cref{lem:nowhere-dense-cliques},
  and let \(\fGameTree\) be the computable function from \cref{lem:game-tree} for \(\C\).

  Let \(k, \ell, q \in \Npos\) and \(p \in \N\) with \(k+\ell+p \leq q\),
  \(\tx \deff (x_1, \dots, x_k)\), \(\ty \deff (y_1, \dots, y_\ell)\),
  \(r \deff \radius(p,q)\),
  let \(\sigma\) be a signature,
  let \(\A\) be a \(\sigma\)-structure with \(G_\A \in \C\),
  let \(\epsilon \in \Qpos\),
  let \(\psi(\tx, \ty) \in \FOplusSigmaPQ\) be \(r\)-local,
  let \(G \in \G_{k+\ell}\) be a connected graph,
  and let \(\phi(\tx, \ty) \deff \delta_{G, r}(\tx, \ty) \land \psi(\tx, \ty)\).
  We set \(r' \deff r \cdot (k+\ell)\),
  and we compute an \(\epsilon' \in \Qpos\)
  such that \(0 < \epsilon' \leq (1+\epsilon)^{1/(t(2r')+1)} - 1\).
  Note that this implies that
  \((1+\epsilon')^{t(2r')+1} \leq 1+\epsilon\).
  Here, to compute \(\epsilon'\),
  we use that \(t\) is computable,
  so that \(\C\) is \emph{effectively} nowhere dense.
  If \(\C\) is nowhere dense but not effectively nowhere dense,
  our results still hold for every fixed \(\phi\) and \(\epsilon\),
  albeit only in a non-uniform sense.

  As described at the beginning of this section,
  we transform all \(\FOplus\) formulae into a normal form.
  Doing this for the input formula \(\phi\)
  takes time \(f_{\textup{normal}}(\phi)\)
  for a computable function \(f_{\textup{normal}}\).
  In the following, we may assume that all formulae we consider
  have been translated into the normal form.

  First, we describe the computations in the preprocessing phase and analyse the running time.
  Afterwards, we describe how to answer queries efficiently,
  and we show that the computations are correct.

  In the preprocessing phase,
  we use \cref{lem:game-tree} on input \(r'\), \(G_\A\), \(\epsilon\), \(\frac{\epsilon'}{2}\)
  to compute a \((\lambda(2r'), 2r')\)-game tree for \(G_\A\)
  in time \(\fGameTree(r', \epsilon, \frac{\epsilon'}{2}) \cdot \abs{A}^{1+\epsilon}\).
  By \cref{lem:game-tree},
  the resulting tree has height at most \(t(2r')\),
  and the distance-\(r'\) neighbourhood cover \(\NCtuple{_{\tCentres}}\)
  at every inner node \(\tCentres\) of the tree
  has overlap at most \(\fNC(r', \frac{\epsilon'}{2}) \cdot \abs{V(G_{\tCentres})}^{\epsilon'/2}\).

  With every inner node \(\tCentres\) of the game tree,
  we associate a signature \(\sigma_{\tCentres} \supseteq \sigma\),
  a set of \(r\)-local \(\FOplusParams{\sigma_{\tCentres}}{p}{q}\) formulae \(\Phi_{\tCentres}\),
  a \(\sigma_{\tCentres}\)-structure \(\A_{\tCentres}\)
  with \(A_{\tCentres} = V(G_{\tCentres})\)
  and \(G_{\A_{\tCentres}} \subseteq G_{\tCentres}\),
  and we describe below how to build a data structure for computing
  \(\bigabs{\gamma(\A_{\tCentres}, \tv)}\) in fixed-parameter-constant time
  (possibly by recursively using the data structures of child nodes)
  given a formula \(\gamma(\tx', \ty') \in \Phi_{\tCentres}\)
  and a tuple \(\tv \in A_{\tCentres}^{\smallabs{\tx'}}\).
  For the root \(()\), we set \(\sigma_{()} \deff \sigma\),
  \(\Phi_{()} \deff \set{\phi}\)
  and \(\A_{()} \deff \A\).
  For every node \(\tCentres\) and every formula \(\gamma(\tx', \ty') \in \Phi_{\tCentres}\),
  there will be numbers \(k' \in [k]\), \(\ell' \in [\ell]\),
  a connected graph \(G' \in \G_{k'+\ell'}\),
  and an \(r\)-local formula \(\psi'(\tx', \ty') \in \FOplusParams{\sigma_{\tCentres}}{p}{q}\)
  such that \(\tx' = (x_1, \dots, x_{k'})\), \(\ty' = (y_1, \dots, y_{\ell'})\),
  and \(\gamma(\tx', \ty') = \delta_{G', r}(\tx', \ty') \land \psi'(\tx', \ty')\).
  By the choice of \(r\),
  \(\delta_{G', r}(\tx', \ty') \in \FOplusParams{\sigma_{\tCentres}}{p}{q}\)
  is \(r\)-local,
  and \(\psi'(\tx', \ty')\) is \(r\)-local by assumption,
  hence \(\gamma(\tx', \ty') \in \FOplusParams{\sigma_{\tCentres}}{p}{q}\)
  is \(r\)-local.

  In a top-down approach, starting at the root \(()\), we do the following.
  Let \(\tCentres\) be the currently considered node of the game tree.
  First, for every child \(\tCentres \centre\) of \(\tCentres\)
  (and hence, for every centre \(\centre \in \centres_{\tCentres}\)),
  we compute \(\A_{\tCentres}[\centreToCluster_{\tCentres}(\centre)]\)
  from \(\A_{\tCentres}\).
  \begin{claim}
    There is a computable function \(\fSub\)
    such that this runs in total time at most
    \(\fSub(r', \sigma_{\tCentres}, \epsilon') \cdot \abs{A_{\tCentres}}^{1+\epsilon'}\).
  \end{claim}
  \begin{claimproof}
    First, for every element \(v \in A_{\tCentres}\)
    and every \(R \in \sigma_{\tCentres}\) of arity at least \(1\),
    we compute the list \(L(v, R)\)
    of all tuples \(\tw \in R(\A_{\tCentres})\)
    where \(w_1 = v\).
    For this, we iterate over all \(R \in \sigma_{\tCentres}\) with \(\ar(R) \geq 1\)
    and over all \(\tw \in R(\A_{\tCentres})\),
    and we add \(\tw\) to \(L(w_1, R)\).
    This takes time \(\bigO(\norm{\A_{\tCentres}})\).

    Then, to compute \(\A_{\tCentres}[\centreToCluster_{\tCentres}(\centre)]\),
    we iterate over all \(R \in \sigma_{\tCentres}\) with \(\ar(R) \geq 1\),
    all \(v \in \centreToCluster_{\tCentres}(\centre)\),
    and all \(\tw \in L(v, R)\),
    and we add \(\tw\) to \(R(\A_{\tCentres}[\centreToCluster_{\tCentres}(\centre)])\)
    if and only if \(\tw\) is fully contained in \(\centreToCluster_{\tCentres}(\centre)\),
    that is, \(\tilde{w} \subseteq \centreToCluster_{\tCentres}(\centre)\),
    which we can check in time
    \(\bigO\bigl(\ar(R) \cdot \fNC(r', \frac{\epsilon'}{2})\bigr)\)
    by \cref{thm:neighbourhood-cover}.
    For all \(0\)-ary \(R \in \sigma_{\tCentres}\),
    we simply set
    \(R(\A_{\tCentres}[\centreToCluster_{\tCentres}(\centre)]) \deff R(\A_{\tCentres})\).
    When doing this for all \(\centre \in \centres_{\tCentres}\),
    since \(\NC_{\tCentres}\) has overlap at most
    \(\fNC(r', \frac{\epsilon'}{2}) \cdot \abs{A_{\tCentres}}^{\epsilon'/2}\),
    every list \(L(v,R)\) is considered at most
    \(\fNC(r', \frac{\epsilon'}{2}) \cdot \abs{A_{\tCentres}}^{\epsilon'/2}\)
    times,
    so the same applies to every single tuple in every relation.
    Hence, after the lists \(L(v,R)\) have been computed,
    the running time for computing
    \(\A_{\tCentres}[\centreToCluster_{\tCentres}(\centre)]\)
    for all \(\centre \in \centres_{\tCentres}\)
    can be upper bounded by
    \(\bigO\bigl(
      \norm{\A_{\tCentres}}
      \cdot
      \fNC(r', \frac{\epsilon'}{2}) \cdot \abs{A_{\tCentres}}^{\epsilon'/2}
      \cdot
      \max(a,1) \cdot \fNC(r', \frac{\epsilon'}{2})
    \bigr)\),
    where \(a\) is the maximum arity of a relation symbol
    in \(\sigma_{\tCentres}\).

    By \cref{lem:nowhere-dense-cliques}, we have
    \(\norm{\A_{\tCentres}} \leq \abs{\sigma_{\tCentres}}
    \cdot \fSize(a, \frac{\epsilon'}{2})
    \cdot \abs{A_{\tCentres}}^{1+\epsilon'/2}\).
    Thus, the overall running time can be upper bounded by
    \(\bigO\Bigl(
      \bigl(\fNC(r', \frac{\epsilon'}{2})\bigr)^2 \cdot \max(a,1)
      \cdot \abs{\sigma_{\tCentres}} \cdot \fSize(a, \frac{\epsilon'}{2})
      \cdot \abs{A_{\tCentres}}^{1+\epsilon'}
    \Bigr)\)
    or, alternatively,
    by \(\fSub(r', \sigma_{\tCentres}, \epsilon') \cdot \abs{A_{\tCentres}}^{1+\epsilon'}\)
    for a computable function \(\fSub\).
  \end{claimproof}

  Next, for every child \(\tCentres \centre\) of \(\tCentres\)
  and for the corresponding cluster
  \(\cluster \deff \centreToCluster_{\tCentres}(\centre)\),
  we build a data structure for computing
  \(\bigabs{\gamma(\A_{\tCentres}[\cluster], \tv)}\)
  in fixed-parameter-constant time
  (again, possibly by recursively using the data structures of child nodes)
  given a formula \(\gamma(\tx', \ty') \in \Phi_{\tCentres}\)
  and a tuple \(\tv \in \cluster^{\smallabs{\tx'}}\)
  with \(\vertexToCentre_{\tCentres}(v_1) = \centre\),
  that is, where the vertex \(v_1\) is assigned to the centre \(\centre\).
  For every such combination, since
  \(\neighb{(k'+\ell')r}{\A_{\tCentres}}{v_1}
  \subseteq \neighb{r'}{\A_{\tCentres}}{v_1} \subseteq \cluster\),
  \cref{lem:delta-connected-local} implies that
  \begin{equation}
    \label{eq:global-to-cluster}
    \gamma(\A_{\tCentres}, \tv)
    = \setc{\tw \in A_{\tCentres}^{\ell'}}{\A_{\tCentres} \models \gamma(\tv, \tw)}
    = \setc{\tw \in \cluster^{\ell'}}{\A_{\tCentres}[\cluster] \models \gamma(\tv, \tw)}
    = \gamma(\A_{\tCentres}[\cluster], \tv).
  \end{equation}

  If \(\tCentres \centre\) is a leaf,
  then \(\abs{\cluster} = \abs{W_{\tCentres \centre}} \leq \lambda(2r')\).
  In this case, for every \(r\)-local
  \(\FOplusParams{\sigma_{\tCentres}}{p}{q}\) formula
  \(\gamma(\tx', \ty') \in \Phi_{\tCentres}\),
  we iterate over all
  \(\tv \in \cluster^{\smallabs{\tx'}}\)
  with \(\vertexToCentre_{\tCentres}(v_1) = \centre\),
  and we compute
  \(\bigabs{\gamma(\A_{\tCentres}[\cluster], \tv)}\)
  by brute force.
  Due to the bound on the size of \(\cluster\),
  there is a computable function \(\fLeaf\) such that these computations
  take time at most \(\fLeaf(q, \sigma_{\tCentres})\).

  Now suppose \(\tCentres \centre\) is not a leaf.
  Let \(s \deff \abs{W_{\tCentres \centre}} \leq \lambda(2r')\),
  let \(\set{d_1, \dots, d_s} \deff W_{\tCentres \centre}\)
  with \(d_1 < d_2 < \cdots < d_s\),
  and note that \(\abs{\cluster} > s\)
  (since \(W_{\tCentres \centre} \subsetneq V(H_{\tCentres \centre}) = K\)
  for non-leaves).
  We compute
  \(\A_{\tCentres \centre} \deff
  \A_{\tCentres}[\cluster] \removaldistr d_1 \cdots d_s\)
  from \(\A_{\tCentres}[\cluster]\)
  in time
  \(\fRemoval^*(r, \sigma_{\tCentres}, s, \frac{\epsilon'}{2}) \cdot
  \abs{\cluster}^{1+\epsilon'/2}\)
  for a computable function \(\fRemoval^*\)
  (an iterated version of \(\fRemoval\) combined with \(\fSize\)), and we set
  \(\sigma_{\tCentres \centre} \deff \sigma(\A_{\tCentres \centre})
  = \removalsigrs{\sigma_{\tCentres}}\).
  Note that, as required,
  \(G_{\A_{\tCentres \centre}}
  = G_{\A_{\tCentres}[\cluster]}[\cluster \setminus W_{\tCentres \centre}]
  \subseteq G_{\A_{\tCentres}}[\cluster \setminus W_{\tCentres \centre}]
  \subseteq G_{\tCentres}[\cluster \setminus W_{\tCentres \centre}]
  = G_{\tCentres \centre}\).
  We iterate over all formulae \(\gamma(\tz') \in \Phi_{\tCentres}\).
  For every such formula, as mentioned above,
  there are \(k' \in [k]\) and \(\ell' \in [\ell]\)
  such that \(\tz' = \tx' \ty'\),
  where \(\tx' = (x_1, \dots, x_{k'})\)
  and \(\ty' = (y_1, \dots, y_{\ell'})\).
  Furthermore, we have that
  \(\gamma(\tz') \in \FOplusParams{\sigma_{\tCentres}}{p}{q}\) is \(r\)-local.
  We iterate over all functions \(\mu \colon [k' + \ell'] \to [0,s]\).
  For every such combination of a formula \(\gamma(\tz')\) and function \(\mu\),
  we apply \cref{cor:removal} to the input \(\gamma(\tz')\) and \(\mu\)
  and compute the \(\FOplusParams{\sigma_{\tCentres \centre}}{p}{q}\) formula
  \(\gamma_\mu(\tz'_\mu)\).
  Then, for all \(\tu \in \cluster^{k' + \ell'}\)
  with \(u_i \not\in \set{d_1, \dots, d_s}\) if \(\mu(i) = 0\)
  and \(u_i = d_{\mu(i)}\) else,
  we have \(\A_{\tCentres}[\cluster] \models \gamma(\tu)\)
  if and only if \(\A_{\tCentres \centre} \models \gamma_\mu(\tu_\mu)\).

  For every \(\tv \in \cluster^{k'}\),
  we let \(\CM_{\tv}\) be the set of all \(\mu \colon [k' + \ell'] \to [0,s]\)
  where, for all \(i \in [k']\),
  we have \(v_i \not\in \set{d_1, \dots, d_s}\) if \(\mu(i) = 0\)
  and \(v_i = d_{\mu(i)}\) else.
  Furthermore, for every tuple \(\tu \in \cluster^{k' + \ell'}\),
  let \(\mu_{\tu} \colon [k' + \ell'] \to [0,s]\) be the mapping with
  \(u_i \not\in \set{d_1, \dots, d_s}\) if \(\mu_{\tu}(i) = 0\)
  and \(u_i = d_{\mu_{\tu}(i)}\) else,
  for all \(i \in [k' + \ell']\).
  Hence, \(\CM_{\tv} = \setc{\mu_{\tv \tw}}{\tw \in \cluster^{\ell'}}\).
  \Cref{cor:removal} implies that
  \begin{equation}
    \label{eq:gamma-to-gamma-mu}
    \begin{aligned}
      \bigabs{\gamma(\A_{\tCentres}[\cluster], \tv)}
      &=
      \bigabs{\setc{\tw \in \cluster^{\ell'}}{\A_{\tCentres}[\cluster] \models \gamma(\tv, \tw)}}\\
      &=
      \sum_{\mu \in \CM_{\tv}}
      \bigabs{\setc{\tw \in \cluster^{\ell'}}{\A_{\tCentres}[\cluster] \models \gamma(\tv, \tw)
      \text{ and } \mu = \mu_{\tv \tw}}}\\
      &=
      \sum_{\mu \in \CM_{\tv}}
      \bigabs{\setc{\tw \in \cluster^{\ell'}}{\A_{\tCentres \centre} \models \gamma_\mu(\tu_\mu)
      \text{ and } \mu = \mu_{\tu} \text{, where } \tu \deff \tv \tw}}\\
      &= \sum_{\mu \in \CM_{\tv}} \bigabs{\gamma_\mu(\A_{\tCentres \centre}, \tv_\mu)}.
    \end{aligned}
  \end{equation}
  For the last step, note that, for fixed \(\mu \in \CM_{\tv}\),
  the condition \(\mu = \mu_{\tu}\) implies that \(w_i = d_{\mu(k'+i)}\)
  for all \(i \in [\ell']\) such that \(\mu(k'+i) \neq 0\).
  Hence, mapping \(\tw\) to its entries at the remaining positions
  (those with \(\mu(k'+i) = 0\),
  which range over \(K \setminus W_{\tCentres \centre} = A_{\tCentres \centre}\))
  is a bijection between the sets in the third and fourth line,
  proving the last equality.

  For every \(\mu \colon [k' + \ell'] \to [0,s]\),
  let \(k_\mu \deff \abs{\mu^{-1}(0) \cap [k']}\)
  and
  \(\ell_\mu \deff \abs{\mu^{-1}(0) \cap [k' + 1, k' + \ell']}\).
If \(k_\mu = 0\),
  then we compute
  \(\bigabs{\gamma_\mu(\A_{\tCentres \centre}, \tv_\mu)}
  = \bigabs{\gamma_\mu(\A_{\tCentres \centre})}\)
  via \cref{thm:fo-counting-static} in time
  \(\fCount(\gamma_\mu, \sigma_{\tCentres \centre}, \frac{\epsilon'}{2})
  \cdot \abs{\cluster}^{1+\epsilon'/2}\).
  Otherwise, if \(k_\mu > 0\),
  we apply \cref{thm:rank-preserving} to \(\sigma_{\tCentres \centre}\),
  \(p\), \(q\), and \(\gamma_\mu(\tz'_\mu)\)
  (possibly by renaming variables
  and adding \(z = z\) for unused variables \(z\)
  to bring \(\gamma_\mu(\tz'_\mu)\) into the right format),
  and, for every graph \(G' \in \G_\mu \deff \G_{k_\mu + \ell_\mu}\),
  we obtain a number \(m_{\mu, G'} \in \N\),
  an \(\FO[\sigma_{\tCentres \centre}]\) sentence \(\xi^i_{\mu, G'}\)
  for every \(i \in [m_{\mu,G'}]\),
  and an \(r\)-local formula
  \(\psi^i_{\mu, G', I}\bigl((\tz'_\mu)_I\bigr)
  \in \FOplusParams{\sigma_{\tCentres \centre}}{p}{q}\)
  for every \(i \in [m_{\mu, G'}]\)
  and every connected component \(I\) of \(G'\).
  Using \cref{thm:fo-testing-enumeration},
  in time \(\fFOTE(\xi^i_{\mu, G'}, \sigma_{\tCentres \centre}, \frac{\epsilon'}{2})
  \cdot \abs{\cluster}^{1+\epsilon'/2}\)
  for every \(i \in [m_{\mu, G'}]\),
  we check whether \(\A_{\tCentres \centre} \models \xi^i_{\mu, G'}\),
  and we let \(J_{\mu, G'} \subseteq [m_{\mu, G'}]\)
  be the set of all indices where this is the case.
  For every \(j \in J_{\mu, G'}\), we let
  \[\gamma^j_{\mu, G'}(\tz'_\mu) \deff \delta_{G', r}(\tz'_\mu)
  \land \Land_{I \in \components(G')} \psi^j_{\mu, G', I}\bigl((\tz'_\mu)_I\bigr)
  \in \FOplusParams{\sigma_{\tCentres \centre}}{p}{q},\]
  which is \(r\)-local,
  and we set
  \(\gamma_{\mu, G'}(\tz'_\mu) \deff \delta_{G', r}(\tz'_\mu) \land \gamma_\mu(\tz'_\mu)\).
  By \cref{thm:rank-preserving},
  for every \(\tu \in A_{\tCentres \centre}^{k_\mu + \ell_\mu}\),
  it holds that
  \(
    \A_{\tCentres \centre} \models \gamma_{\mu, G'}(\tu)
    \iff
    \A_{\tCentres \centre} \models \Lor_{j \in J_{\mu, G'}} \gamma^j_{\mu, G'}(\tu),
  \)
  and the \(\gamma^j_{\mu, G'}\) are mutually exclusive,
  that is,
  there is at most one \(j \in J_{\mu, G'}\)
  where
  \(\A_{\tCentres \centre} \models \gamma^j_{\mu, G'}(\tu)\).
  This shows that,
  for all \(\tv' \in A_{\tCentres \centre}^{k_\mu}\),
  we have
  \begin{equation}
    \label{eq:gamma-mu-to-gamma-j-mu}
    \bigabs{\gamma_\mu(\A_{\tCentres \centre}, \tv')}
    = \sum_{G' \in \G_\mu} \bigabs{\gamma_{\mu, G'}(\A_{\tCentres \centre}, \tv')}
    = \sum_{G' \in \G_\mu} \sum_{j \in J_{\mu, G'}}
      \bigabs{\gamma^j_{\mu, G'}(\A_{\tCentres \centre}, \tv')}.
  \end{equation}
  Finally, we apply \cref{lem:reduction-connected} to every \(G' \in \G_\mu\)
  and obtain a number \(m_{G'} \in \N\),
  connected graphs \(G'_1, \dots, G'_{m_{G'}}\),
  and a polynomial \(P_{G'}(X_1, \dots, X_{m_{G'}}) \in \Z[X_1, \dots, X_{m_{G'}}]\).
For all \(i \in [m_{G'}]\) and \(j \in J_{\mu, G'}\),
  let
  \[
    \gamma^j_{\mu, G', i}\bigl((\tz'_\mu)_{V(G'_i)}\bigr)
    \deff \delta_{G'_i, r}\bigl((\tz'_\mu)_{V(G'_i)}\bigr)
    \land \Land_{\substack{I \in \components(G'),\\ I \subseteq V(G'_i)}}
    \psi^j_{\mu, G', I}\bigl((\tz'_\mu)_I\bigr)
    \in \FOplusParams{\sigma_{\tCentres \centre}}{p}{q},
  \]
  which is \(r\)-local.
  By \cref{lem:reduction-connected},
  for every \(\tv' \in A_{\tCentres \centre}^{k_\mu}\),
  it holds that
  \begin{equation}
    \label{eq:gamma-j-mu-to-connected}
    \bigabs{\gamma^j_{\mu, G'}(\A_{\tCentres \centre}, \tv')}
    = P_{G'}\Bigl(
      \bigabs{\gamma^j_{\mu, G', 1}(\A_{\tCentres \centre}, \tv'_{V(G'_1)})},
      \dots,
      \bigabs{\gamma^j_{\mu, G', m_{G'}}(\A_{\tCentres \centre}, \tv'_{V(G'_{m_{G'}})})}
    \Bigr).
  \end{equation}
  If \(V(G'_i) \cap [k_\mu] = \emptyset\),
  then we use \cref{thm:fo-counting-static} to compute
  \(\bigabs{\gamma^j_{\mu, G', i}(\A_{\tCentres \centre}, \tv'_{V(G'_i)})}
  = \bigabs{\gamma^j_{\mu, G', i}(\A_{\tCentres \centre})}\)
  in time \(\fCount(\gamma^j_{\mu, G', i}, \sigma_{\tCentres \centre}, \frac{\epsilon'}{2})
  \cdot \abs{\cluster}^{1+\epsilon'/2}\).
  If \(V(G'_i) \subseteq [k_\mu]\),
  then we have
  \(\bigabs{\gamma^j_{\mu, G', i}(\A_{\tCentres \centre}, \tv'_{V(G'_i)})} \in \set{0,1}\)
  with \(\bigabs{\gamma^j_{\mu, G', i}(\A_{\tCentres \centre}, \tv'_{V(G'_i)})} = 1\)
  if and only if \(\A_{\tCentres \centre} \models \gamma^j_{\mu, G', i}(\tv'_{V(G'_i)})\).
  In order to prepare the efficient computation of
  \(\bigabs{\gamma^j_{\mu, G', i}(\A_{\tCentres \centre}, \tv'_{V(G'_i)})}\),
  we run the preprocessing of \cref{thm:fo-testing-enumeration}
  on \(\A_{\tCentres \centre}\) and \(\gamma^j_{\mu, G', i}\bigl((\tz'_\mu)_{V(G'_i)}\bigr)\)
  in time \(\fFOTE(\gamma^j_{\mu, G', i}, \sigma_{\tCentres \centre}, \frac{\epsilon'}{2})
  \cdot \abs{\cluster}^{1+\epsilon'/2}\).
  In the remaining case, where \(V(G'_i) \cap [k_\mu] \neq \emptyset\)
  and \(V(G'_i) \not\subseteq [k_\mu]\),
  we add \(\gamma^j_{\mu, G', i}\bigl((\tz'_\mu)_{V(G'_i)}\bigr)\)
  to \(\Phi_{\tCentres \centre}\)
  after renaming the variables in the formula and the vertices of the graph \(G'_i\)
  such that the resulting formula is of the right format.
After running the described procedure for all \(\gamma(\tz') \in \Phi_{\tCentres}\),
  \(\mu \colon [k'+\ell'] \to [0,s]\),
  and \(G' \in \G_\mu\),
  we continue the computation with \(\tCentres \centre\) instead of \(\tCentres\).

  This completes the preprocessing.
  The game tree has been computed in time
  \(\fGameTree(r', \epsilon, \frac{\epsilon'}{2}) \cdot \abs{A}^{1+\epsilon}\).
  Since the computed game tree is a \((\lambda(2r'), 2r')\)-game tree,
  the sum of \(\abs{W_{\tCentres \centre}}\) over all steps
  along any root-to-leaf path in the game tree is at most \(\lambda(2r')\),
  so all formulae considered in the algorithm are
  \(\FOplusParams{\removalsig{\sigma}{r,\lambda(2r')}}{p}{q}\) formulae.
  As described at the beginning of this section,
  the number of formulae we need to consider is finitely bounded,
  and the bound can be computed from
  \(\removalsig{\sigma}{r,\lambda(2r')}\) and \(q\).
  Hence, there is a computable function \(\fNode\) such that,
  for every inner node \(\tCentres\) of the game tree,
  the preprocessing
  (excluding the preprocessing for the children of \(\tCentres\))
  takes time at most
  \(\fNode(q, \removalsig{\sigma}{r,\lambda(2r')}, \epsilon')
  \cdot \abs{A_{\tCentres}}^{1+\epsilon'}\),
  where we use that
  \(
  \sum_{\centre \in \centres_{\tCentres}}
  \abs{\centreToCluster_{\tCentres}(\centre)}^{1+\epsilon'/2}
  \leq \abs{A_{\tCentres}}^{\epsilon'/2} \cdot
  \sum_{\centre \in \centres_{\tCentres}} \abs{\centreToCluster_{\tCentres}(\centre)}
  \leq \fNC(r', \frac{\epsilon'}{2}) \cdot \abs{A_{\tCentres}}^{1+\epsilon'}
  \).
  Analogously to the proof of \cref{lem:game-tree} and by the choice of \(\epsilon'\),
  it can be shown that there is a computable function \(f_{\textup{pre}}\)
  such that the whole preprocessing procedure takes time at most
  \(f_{\textup{pre}}(q, \sigma, \epsilon) \cdot \abs{A}^{1+\epsilon}\).

  After the preprocessing, whenever we are given a node \(\tCentres\) of the game tree,
  a formula \(\gamma(\tx', \ty') \in \Phi_{\tCentres}\),
  and a tuple \(\tv \in A_{\tCentres}^{\smallabs{\tx'}}\)
  (this includes the case \(\tCentres = ()\),
  \(\gamma(\tx', \ty') = \phi(\tx, \ty)\),
  and \(\tv \in A^k\)),
  we recursively compute \(\bigabs{\gamma(\A_{\tCentres}, \tv)}\) as follows.
  Let \(k' \deff \abs{\tx'}\), \(\ell' \deff \abs{\ty'}\),
  \(\tz' \deff \tx'\ty'\),
  and \(\centre \deff \vertexToCentre_{\tCentres}(v_1)\).
  In time \(\bigO\bigl(k' \cdot \fNC(r', \frac{\epsilon'}{2})\bigr)\),
  using \cref{thm:neighbourhood-cover},
  we check whether
  \(\tilde{v} \subseteq \centreToCluster_{\tCentres}(\centre)\).
  If this is not the case, then, by \cref{lem:delta-connected-local}
  (using \(G_{\A_{\tCentres}} \subseteq G_{\tCentres}\), which implies
  \(\neighb{(k' + \ell') r}{\A_{\tCentres}}{v_1}
  \subseteq \neighb{(k' + \ell') r}{G_{\tCentres}}{v_1}
  \subseteq \neighb{r'}{G_{\tCentres}}{v_1}
  \subseteq \centreToCluster_{\tCentres}(\centre)\)),
  there is no \(\tw \in A_{\tCentres}^{\ell'}\)
  with \(\A_{\tCentres} \models \gamma(\tv, \tw)\).
  Thus, we return \(0\).
Otherwise, again by \cref{lem:delta-connected-local}
  (and as stated in \cref{eq:global-to-cluster}),
  we have \(\bigabs{\gamma(\A_{\tCentres}, \tv)}
  = \bigabs{\gamma(\A_{\tCentres}[\centreToCluster_{\tCentres}(\centre)], \tv)}\).
  If \(\tCentres \centre\) is a leaf,
  then we have already precomputed
  \(\bigabs{\gamma(\A_{\tCentres}[\centreToCluster_{\tCentres}(\centre)], \tv)}\),
  and we simply return this precomputed value.
  If \(\tCentres \centre\) is not a leaf,
  then we let \(s \deff \abs{W_{\tCentres \centre}}\),
  \(\set{d_1, \dots, d_s} \deff W_{\tCentres \centre}\)
  with \(d_1 < d_2 < \cdots < d_s\),
  and we let \(\CM_{\tv}\) be the set of all \(\mu \colon [k'+\ell'] \to [0,s]\)
  where, for all \(i \in [k']\),
  we have \(v_i \not\in \set{d_1, \dots, d_s}\) if \(\mu(i) = 0\)
  and \(v_i = d_{\mu(i)}\) else.

  If \(\tilde{v} \subseteq W_{\tCentres \centre}\),
  then \(\mu^{-1}(0) \cap [k'] = \emptyset\) for all \(\mu \in \CM_{\tv}\),
  and \(\bigabs{\gamma(\A_{\tCentres}[\centreToCluster_{\tCentres}(\centre)], \tv)}\)
  can be computed by combining the precomputed values
  \(\bigabs{\gamma_\mu(\A_{\tCentres \centre}, \tv_\mu)}
  = \bigabs{\gamma_\mu(\A_{\tCentres \centre})}\)
  for \(\mu \in \CM_{\tv}\) as described in
  \cref{eq:gamma-to-gamma-mu}.

  Otherwise, we have
  \(\mu^{-1}(0) \cap [k'] \neq \emptyset\) for all \(\mu \in \CM_{\tv}\).
  Thus, by combining
  \cref{eq:gamma-to-gamma-mu,eq:gamma-mu-to-gamma-j-mu,eq:gamma-j-mu-to-connected},
  we can compute
  \(\bigabs{\gamma(\A_{\tCentres}[\centreToCluster_{\tCentres}(\centre)], \tv)}\)
  from the values
  \(\bigabs{\gamma^j_{\mu, G', i}(\A_{\tCentres \centre}, (\tv_\mu)_{V(G'_i)})}\)
  for all \(\mu \in \CM_{\tv}\),
  \(G' \in \G_\mu\),
  \(j \in J_{\mu, G'}\),
  and \(i \in [m_{G'}]\).
  If \(V(G'_i) \cap [k_\mu] = \emptyset\),
  then
  \(\bigabs{\gamma^j_{\mu, G', i}(\A_{\tCentres \centre}, (\tv_\mu)_{V(G'_i)})}
  = \bigabs{\gamma^j_{\mu, G', i}(\A_{\tCentres \centre})}\)
  has already been precomputed.
  If \(V(G'_i) \subseteq [k_\mu]\),
  then
  \(\bigabs{\gamma^j_{\mu, G', i}(\A_{\tCentres \centre}, (\tv_\mu)_{V(G'_i)})}
  \in \set{0,1}\),
  and this can be computed in time
  \(\fFOTE(\gamma^j_{\mu, G', i}, \sigma_{\tCentres \centre}, \frac{\epsilon'}{2})\)
  based on the preprocessing.
  In the remaining case,
  we have \(V(G'_i) \cap [k_\mu] \neq \emptyset\)
  and \(V(G'_i) \not\subseteq [k_\mu]\),
  which implies that
  \(\gamma^j_{\mu, G', i}((\tz'_\mu)_{V(G'_i)}) \in \Phi_{\tCentres \centre}\)
  (after possibly renaming variables),
  and we can compute the value
  \(\bigabs{\gamma^j_{\mu, G', i}(\A_{\tCentres \centre}, (\tv_\mu)_{V(G'_i)})}\)
  by recursion.
  The depth of the recursion is bounded by the height of the game tree.
  Since the running time for every single step in the recursion can be bounded
  in terms of \(\sigma\), \(q\), and \(\epsilon\),
  and the depth is also bounded in terms of \(q\),
  there is a computable function \(f_{\textup{lookup}}\)
  such that computing \(\bigabs{\phi(\A, \tv)}\) on input \(\tv \in A^k\)
  can be done in time \(f_{\textup{lookup}}(q, \sigma, \epsilon)\).

  Setting
  \(
    f(\phi, q, \sigma, \epsilon) \deff \max \set{
      f_{\textup{normal}}(\phi),
      f_{\textup{pre}}(q, \sigma, \epsilon),
      f_{\textup{lookup}}(q, \sigma, \epsilon)
    }
  \)
  finishes the proof of \cref{lem:foplus-counting-connected}.
\end{proof}

Based on \cref{lem:foplus-counting-connected}, we can now prove \cref{thm:foplus-counting}.

\begin{proof}[Proof of \cref{thm:foplus-counting}]
  If \(\ell = 0\),
  then \(\bigabs{\phi(\A, \tv)} \in \set{0,1}\) for all \(\tv \in A^k\),
  and \(\bigabs{\phi(\A, \tv)} = 1\) if and only if \(\A \models \phi(\tv)\).
  Hence, for this case, we can use the algorithm from \cref{thm:fo-testing-enumeration}
  with \(\fFOTE(\phi, \sigma, \epsilon) \cdot \abs{A}^{1+\epsilon}\) preprocessing
  that can afterwards test whether \(\A \models \phi(\tv)\) for any tuple \(\tv \in A^k\)
  in time \(\fFOTE(\phi, \sigma, \epsilon)\),
  and \(\fFOTE\) is a computable function.

  If \(k = 0\),
  then we only need to compute \(\bigabs{\phi(\A)}\),
  which we can do using \cref{thm:fo-counting-static}
  in time \(\fCount(\phi, \sigma, \epsilon) \cdot \abs{A}^{1+\epsilon}\),
  and \(\fCount\) is a computable function.

  Thus, in the following, assume \(k, \ell \in \Npos\).
  Let \(r \deff \radius(p,q)\),
  \(\tx \deff (x_1, \dots, x_k)\),
  \(\ty \deff (y_1, \dots, y_\ell)\),
  and \(\tz \deff \tx \ty\).
  Without loss of generality,
  we may assume that \(\free(\phi) = \tilde{x} \cup \tilde{y}\).
  For \(G \in \G_{k+\ell}\), let
  \(\phi_G(\tx, \ty) \deff \delta_{G,r}(\tx, \ty) \land \phi(\tx, \ty)\).
  It holds that
  \begin{equation}
    \label{eq:foplus-counting-to-G}
    \bigabs{\phi(\A, \tv)} = \sum_{G \in \G_{k+\ell}} \bigabs{\phi_G(\A, \tv)}.
  \end{equation}
  Hence, it suffices to prove the result for formulae of the form
  \(\phi_G(\tx, \ty) = \delta_{G,r}(\tx, \ty) \land \phi(\tx, \ty)\)
  for a graph \(G \in \G_{k+\ell}\).
For this, we first apply \cref{thm:rank-preserving} to
  \(\sigma\), \(p\), \(q\), and \(\phi\),
  and, for every such \(G\), we obtain a number \(m_G \in \N\),
  an \(\FO[\sigma]\) sentence \(\xi^i_G\) for each \(i \in [m_G]\),
  and, for each \(i \in [m_G]\) and every \(I \in \components(G)\),
  an \(r\)-local formula \(\psi^i_{G,I}(\tz_I) \in \FOplusSigmaPQ\).
  Let \[\phi^i_G(\tz) \deff \delta_{G,r}(\tz) \land \Land_{I \in \components(G)} \psi^i_{G,I}(\tz_I).\]
  For every \(i \in [m_G]\),
  we use \cref{thm:fo-testing-enumeration} to check whether \(\A \models \xi^i_G\),
  and we let \(J_G\) be the set of all \(i\) where this is the case.
  Then, by \cref{thm:rank-preserving},
  we have \(\phi_G(\A) = \biguplus_{j \in J_G} \phi^j_G(\A)\).
  Hence, for every \(\tv \in A^k\), it holds that
  \begin{equation}
    \label{eq:foplus-counting-G-to-j-G}
    \bigabs{\phi_G(\A, \tv)} = \sum_{j \in J_G} \bigabs{\phi^j_G(\A, \tv)}.
  \end{equation}

  Next, we apply \cref{lem:reduction-connected} to \(G\)
  and obtain a number \(s_G \in \N\),
  connected graphs \(G_1, \dots, G_{s_G}\) with \(V_i \deff V(G_i) \subseteq V(G)\)
  for all \(i \in [s_G]\),
  and a polynomial \(P(X_1, \dots, X_{s_G}) \in \Z[X_1, \dots, X_{s_G}]\).
  For \(j \in J_G\) and \(i \in [s_G]\),
  let \[\phi^j_{G,i}(\tz_{V_i}) \deff \delta_{G_i,r}(\tz_{V_i})
  \land \Land_{I \in \components(G), I \subseteq V_i} \psi^j_{G,I}(\tz_I),\]
  which is an \(r\)-local formula in \(\FOplusSigmaPQ\).
Then, for every \(\tv \in A^k\), we have
  \begin{equation}
    \label{eq:foplus-counting-j-G-to-j-G-i}
    \bigabs{\phi^j_G(\A, \tv)}
    = P\Bigl(\bigabs{\phi^j_{G,1}(\A, \tv_{V_1})}, \dots,
    \bigabs{\phi^j_{G,s_G}(\A, \tv_{V_{s_G}})}\Bigr).
  \end{equation}
  If \(V_i \cap [k] = \emptyset\),
  then \(\bigabs{\phi^j_{G,i}(\A, \tv_{V_i})} = \bigabs{\phi^j_{G,i}(\A)}\),
  which we can compute using \cref{thm:fo-counting-static}.
  If \(V_i \subseteq [k]\),
  then \(\bigabs{\phi^j_{G,i}(\A, \tv_{V_i})} \in \set{0,1}\),
  and we run the preprocessing of \cref{thm:fo-testing-enumeration}.
  After the preprocessing, given a tuple \(\tv \in A^k\),
  we can compute \(\bigabs{\phi^j_{G,i}(\A, \tv_{V_i})}\)
  in time \(\fFOTE(\phi^j_{G,i}, \sigma, \epsilon)\).
  Finally, if \(V_i \cap [k] \neq \emptyset\) and \(V_i \not\subseteq [k]\),
  then we run the preprocessing of \cref{lem:foplus-counting-connected}
  (possibly after renaming the variables in \(\phi^j_{G,i}(\tz_{V_i})\)
  and the vertices of \(G_i\)).
  Also in this case, after the preprocessing and given a tuple \(\tv \in A^k\),
  we can compute \(\bigabs{\phi^j_{G,i}(\A, \tv_{V_i})}\)
  in time \(f'(\phi^j_{G,i}, q, \sigma, \epsilon)\),
  where \(f'\) is the computable function from \cref{lem:foplus-counting-connected}.

  All in all, there is a computable function \(f\) such that the preprocessing runs
  in time \(f(\phi, q, \sigma, \epsilon) \cdot \abs{A}^{1+\epsilon}\),
  and, after the preprocessing and given a tuple \(\tv \in A^k\),
  we can compute \(\bigabs{\phi^j_{G,i}(\A, \tv_{V_i})}\)
  for all \(G \in \G_{k + \ell}\), \(j \in J_G\), and \(i \in [s_G]\)
  and compute \(\bigabs{\phi(\A, \tv)}\) based on these results using
  \cref{eq:foplus-counting-to-G,eq:foplus-counting-G-to-j-G,eq:foplus-counting-j-G-to-j-G-i}
  in time \(f(\phi, q, \sigma, \epsilon)\).
\end{proof}
 \section{Query Answering and Enumeration for cgFOC}
\label{sec:cgfoc}

In this section, we use \cref{thm:fo-counting} to prove \cref{thm:answering-enumeration}.
First, we give the definition of the first-order logic with counting \(\FOC\)
from~\cite{KuskeSchweikardt_2017_FOCN} and its fragment
\emph{clique-guarded first-order logic with counting} (\(\cgFOC\))
from~\cite{vanBergeremLangeSchweikardt_2026_cgFOC}.

Let \(\sigma\) be a signature,
and let \(\vars\) be a fixed and countably infinite set of variables.
A \emph{\(\sigma\)-interpretation} \(\I = (\A, \beta)\) consists of a
\(\sigma\)-structure \(\A\) and an \emph{assignment} \(\beta \colon \vars \to A\).
For \(k \in \N\), pairwise distinct variables \(x_1, \dots, x_k \in \vars\),
and elements \(v_1, \dots, v_k \in A\),
we write \(\I \frac{v_1, \dots, v_k}{x_1, \dots, x_k}\)
for the interpretation \((\A, \beta \frac{v_1, \dots, v_k}{x_1, \dots, x_k})\),
where \(\beta \frac{v_1, \dots, v_k}{x_1, \dots, x_k}\) is the assignment \(\beta'\) with
\(\beta'(x_i) = v_i\) for every \(i \in [k]\)
and \(\beta'(y) = \beta(y)\) for all \(y \in \vars \setminus \set{x_1, \dots, x_k}\).

A \emph{numerical predicate collection} is a triple \((\PP, \ar, \sem{.})\)
where \(\PP\) is a countable set of \emph{predicate names},
and, to each \(\Pred \in \PP\),
\(\ar\) assigns an \emph{arity} \(\ar(\Pred) \in \Npos\)
and \(\sem{.}\) assigns a \emph{semantics}
\(\sem{\Pred} \subseteq \Z^{\ar(\Pred)}\).
For the remainder of this paper,
fix a numerical predicate collection \((\PP, \ar, \sem{.})\).
In algorithms, we will assume that machines have access to oracles
for evaluating the numerical predicates in constant time.
That is, given a predicate \(\Pred \in \PP\)
and a tuple \((i_1, \dots, i_{\ar(\Pred)})\) of integers,
answering the oracle call
\enquote{\((i_1, \dots, i_{\ar(\Pred)}) \in \sem{\Pred}\)?}
takes time \(\bigO(1)\).

\begin{definition}[\(\FOC{[}\sigma{]}\)]
  \label[definition]{def:foc}
  The set of \emph{formulae} and \emph{counting terms} for \(\FOC[\sigma]\)
  is built according to the following rules.
\begin{bracketenumerate}
    \item\label{def:fo-atomic}
      \(x_1{=}x_2\) and \(R(x_1, \dots, x_k)\) are formulae for
      \(R \in \sigma\), \(k \deff \ar(R)\),
      and \(x_1, x_2, \dots, x_k \in \vars\).
      In particular, if \(\ar(R) = 0\), then \(R()\) is a formula.
    \item\label{def:fo-bool}
      If \(\phi\) and \(\psi\) are formulae,
      then \(\neg \phi\) and \((\phi \lor \psi)\) are formulae.
    \item\label{def:fo-exists}
      If \(\phi\) is a formula and \(x \in \vars\),
      then \(\exists x\, \phi\) is a formula.
    \item\label{def:foc-countingterm}
      If \(\phi\) is a formula and
      \(\tx = (x_1, \dots, x_k)\) is a tuple of \(k\) pairwise distinct variables
      for some \(k \in \N\),
      then \(\FOCCount{\tx}{\phi}\) is a counting term.
      This includes \(k=0\), so \(\FOCCount{()}{\phi}\) is a counting term.
    \item\label{def:foc-constterm}
      Every integer \(i \in \Z\) is a counting term.
    \item\label{def:foc-plustimesterm}
      If \(t_1\) and \(t_2\) are counting terms,
      then \((t_1 + t_2)\) and \((t_1 \cdot t_2)\) are also counting terms.
    \item\label{def:foc-P}
      If \(\Pred \in \PP\), \(m = \ar(\Pred)\),
      and \(t_1, \dots, t_m\) are counting terms,
      then \(\Pred(t_1, \dots, t_m)\) is a formula.
  \end{bracketenumerate}
Let \(\I = (\A, \beta)\) be a \(\sigma\)-interpretation.
  For a formula or counting term \(\xi\) from \(\FOC[\sigma]\),
  the semantics \(\sem{\xi}^\I\) is defined as follows:
\begin{bracketenumerate}
    \item
      \(\sem{x_1{=}x_2}^\I = 1\) if \(\beta(x_1) = \beta(x_2)\),
      and \(\sem{x_1{=}x_2}^\I = 0\) otherwise;
      \(\sem{R(x_1, \dots, x_k)}^\I = 1\) if
      \(\bigl(\beta(x_1), \dots, \beta(x_k)\bigr) \in R(\A)\), and
      \(\sem{R(x_1, \dots, x_k)}^\I = 0\) otherwise.
    \item
      \(\sem{\neg \phi}^\I = 1 - \sem{\phi}^\I\)
      and \(\sem{(\phi \lor \psi)}^\I = \max \set{\sem{\phi}^\I, \sem{\psi}^\I}\).
    \item
      \(\sem{\exists x\, \phi}^\I = \max \setc{\sem{\phi}^{\I\frac{v}{x}}}{v \in A}\).
    \item \(\sem{\FOCCount{\tx}{\phi}}^\I =
      \bigabs{\bigsetc{(v_1, \dots, v_k) \in A^k}
      {\sem{\phi}^{\I \frac{v_1, \dots, v_k}{x_1, \dots, x_k}} = 1}}\),
      where \(\tx = (x_1, \dots, x_k)\).
    \item \(\sem{i}^\I = i\) for \(i \in \Z\).
    \item \(\sem{(t_1 + t_2)}^\I = \sem{t_1}^\I + \sem{t_2}^\I\)
      and \(\sem{(t_1 \cdot t_2)}^\I = \sem{t_1}^\I \cdot \sem{t_2}^\I\).
    \item \(\sem{\Pred(t_1, \dots, t_m)}^\I = 1\)
      if \((\sem{t_1}^\I, \dots, \sem{t_m}^\I) \in \sem{\Pred}\),
      and \(\sem{\Pred(t_1, \dots, t_m)}^\I = 0\) otherwise.
  \end{bracketenumerate}
\end{definition}

We write \((\phi \land \psi)\), \((\phi \rightarrow \psi)\), and \(\forall x\, \phi\)
as shorthands for \(\neg(\neg \phi \lor \neg \psi)\), \((\neg \phi \lor \psi)\),
and \(\neg \exists x\, \neg \phi\).
For counting terms \(t_1\) and \(t_2\),
we write \((t_1 - t_2)\) as a shorthand for \(\bigl(t_1 + ((-1) \cdot t_2)\bigr)\).
Counting terms \(t\) of the form \(\FOCCount{\tx}{\phi}\)
(i.e., obtained by applying rule~\eqref{def:foc-countingterm})
are also called \emph{\(\#\)-terms}.
An \emph{expression} is a formula or a counting term.
For an expression \(\xi\), we write \(\sigma(\xi)\) for the set of all
relation symbols that occur in \(\xi\).

The \emph{free variables} \(\free(\xi)\) of an expression \(\xi\) are inductively defined as follows:
\eqref{def:fo-atomic}
\(\free(x_1{=}x_2) = \set{x_1, x_2}\)
and \(\free\bigl(R(x_1, \dots, x_k)\bigr) = \set{x_1, \dots, x_k}\);
\eqref{def:fo-bool}
\(\free(\neg \phi) = \free(\phi)\)
and \(\free\bigl((\phi \lor \psi)\bigr) = \free(\phi) \cup \free(\psi)\);
\eqref{def:fo-exists}
\(\free(\exists x\, \phi) = \free(\phi) \setminus \set{x}\);
\eqref{def:foc-countingterm}
\(\free\bigl(\FOCCount{(x_1, \dots, x_k)}{\phi}\bigr)
= \free(\phi) \setminus \set{x_1, \dots, x_k}\);
\eqref{def:foc-constterm}
\(\free(i) = \emptyset\) for \(i \in \Z\);
\eqref{def:foc-plustimesterm}
\(\free\bigl((t_1 + t_2)\bigr) = \free\bigl((t_1 \cdot t_2)\bigr) = \free(t_1) \cup \free(t_2)\);
\eqref{def:foc-P} \(\free\bigl(\Pred(t_1, \dots, t_m)\bigr) = \bigcup_{i=1}^m \free(t_i)\).
We write \(\xi(z_1, \dots, z_k)\)
for pairwise distinct variables \(z_1, \dots, z_k\)
to indicate that \(\free(\xi) \subseteq \set{z_1, \dots, z_k}\).
A \emph{sentence} is a formula without free variables.

For an expression \(\xi(x_1, \dots, x_k)\), a \(\sigma\)-structure \(\A\),
and a tuple \(\tv = (v_1, \dots, v_k) \in A^k\),
we set \(\sem{\xi(\tv)}^\A \deff \sem{\xi}^{(\A, \beta)}\)
for one (and hence every) assignment \(\beta\)
with \(\beta(x_i) = v_i\) for all \(i \in [k]\).

The set of formulae for \(\FO[\sigma]\)
is built according to the rules \eqref{def:fo-atomic}--\eqref{def:fo-exists}.
Based on \cref{def:foc}, the definition of \(\cgFOC\)
from~\cite{vanBergeremLangeSchweikardt_2026_cgFOC} looks as follows.

\begin{definition}[\(\cgFOC{[}\sigma{]}\)]
  \label[definition]{def:cgfoc}
  The set of \emph{formulae} and \emph{counting terms} for \(\cgFOC[\sigma]\)
  (clique-guarded first-order logic with counting) is built according to
  rules~\eqref{def:fo-atomic}--\eqref{def:foc-plustimesterm} of \cref{def:foc}
  and the following modified version of rule~\eqref{def:foc-P}.
  \begin{enumerate}[(cg1)]
    \setcounter{enumi}{6}
    \renewcommand{\labelenumi}{\textbf{(\theenumi)}}
    \renewcommand{\theenumi}{cg\arabic{enumi}}
    \item\label{def:cgfoc-P}
      If \(\Pred \in \PP\), \(m = \ar(\Pred)\),
      \(k \in \N\),
      \(R_1, \dots, R_k \in \sigma\),
      \(\tz_i \in \vars^{\ar(R_i)}\) for all \(i \in [k]\),
      and \(t_1, \dots, t_m\) are counting terms
      such that for all \(z, z' \in \bigcup_{p=1}^m \free(t_p)\) with \(z \neq z'\),
      there is some \(i \in [k]\) with \(z, z' \in \tilde{z}_i\),
      then \(\bigl(\Land_{i=1}^k R_i(\tz_i) \land \Pred(t_1, \dots, t_m)\bigr)\) is a formula.
  \end{enumerate}
\end{definition}

Intuitively, this means that every pair of free variables
of a subformula of the form \(\Pred(t_1, \dots, t_m)\) is guarded by an atom.
Hence, the set of vertices that the free variables of \(\Pred(t_1, \dots, t_m)\)
are interpreted with needs to be a clique in the Gaifman graph of the structure.

By \(\FOC\), we denote the union of all \(\FOC[\sigma]\) for arbitrary signatures \(\sigma\).
This applies analogously to \(\cgFOC\) and \(\FO\).

The following lemma provides an efficient reduction from \(\cgFOC\) formulae to \(\FO\) formulae.
It generalises \cite[Lemma~8]{vanBergeremLangeSchweikardt_2026_cgFOC}
from classes of locally bounded expansion to nowhere dense classes.

\begin{lemma}
  \label{lem:cgfoc-preprocessing}
  For every signature \(\sigma\) and every \(\cgFOC[\sigma]\) formula \(\phi(\tx)\),
  there is a signature \(\sigma_\phi \supseteq \sigma\)
  and an \(\FO[\sigma_\phi]\) formula \(\phi'(\tx)\)
  that are computable from \(\sigma\) and \(\phi(\tx)\)
  such that for every effectively nowhere dense graph class \(\C\),
  there is a computable function \(\fPre\) and an algorithm that does the following.
  Given a \(\sigma\)-structure \(\A\) with \(G_\A \in \C\)
  and an \(\epsilon \in \Qpos\),
  the algorithm computes a \(\sigma_\phi\)-expansion \(\A_\phi\) of \(\A\)
  with \( G_{\A_\phi} = G_\A\) such that,
  for all \(\tv \in A^{\abs{\tx}}\),
  it holds that \(\A \models \phi(\tv)\) if and only if \(\A_\phi \models \phi'(\tv)\).
  Furthermore, on input \(\A, \epsilon\), the algorithm runs in time
  \(\fPre(\phi, \sigma, \epsilon) \cdot \abs{A}^{1+\epsilon}\).
\end{lemma}
\begin{proof}
  Let \(\fCountQA\) be the computable function from \cref{thm:fo-counting} for \(\C\).
  We prove the result by structural induction on \(\phi\).
  If \(\phi\) is built without the use of rule~\eqref{def:cgfoc-P} of \cref{def:cgfoc},
  then \(\phi\) is an \(\FO[\sigma]\) formula,
  and we can simply set \(\sigma_\phi \deff \sigma\),
  \(\phi' \deff \phi\), and \(\A_\phi \deff \A\).

  If \(\phi\) is of the form \(\neg \psi\),
  then we apply the induction hypothesis to \(\psi\) to obtain
  a signature \(\sigma_\psi \supseteq \sigma\),
  an \(\FO[\sigma_\psi]\) formula \(\psi'(\tx)\),
  and a \(\sigma_\psi\)-expansion \(\A_\psi\) of \(\A\)
  with the desired properties.
  We set \(\sigma_\phi \deff \sigma_\psi\),
  \(\phi' \deff \neg \psi'\),
  and \(\A_\phi \deff \A_\psi\).

  If \(\phi\) is of the form \(\psi_1 \lor \psi_2\),
  then we apply the induction hypothesis to \(\psi_1\) and \(\psi_2\) to obtain
  signatures \(\sigma_{\psi_1}, \sigma_{\psi_2} \supseteq \sigma\),
  formulae \(\psi_1'(\tx), \psi_2'(\tx)\),
  and structures \(\A_{\psi_1}, \A_{\psi_2}\)
  with the desired properties.
  Without loss of generality,
  by possibly renaming the relation symbols in \(\sigma_{\psi_2} \setminus \sigma\),
  we may assume that \(\sigma_{\psi_1} \cap \sigma_{\psi_2} = \sigma\).
  We set \(\sigma_\phi \deff \sigma_{\psi_1} \cup \sigma_{\psi_2}\),
  \(\phi' \deff \psi'_1 \lor \psi'_2\),
  and we let \(\A_\phi\) be the \(\sigma_\phi\)-expansion
  of both \(\A_{\psi_1}\) and \(\A_{\psi_2}\).

  If \(\phi\) is of the form \(\exists x\, \psi\),
  we proceed analogously to the case \(\neg \psi\),
  and we set \(\sigma_\phi \deff \sigma_\psi\),
  \(\phi' \deff \exists x\, \psi'\),
  and \(\A_\phi \deff \A_\psi\).

  In all these cases, the computation of \(\A_\phi\)
  trivially runs in time linear in \(\norm{\A_\phi}\),
  based on the recursively computed structure(s) for the respective subformula(e).
  Furthermore, by \cref{lem:nowhere-dense-cliques},
  \(\norm{\A_\phi} \leq \smallabs{\sigma_\phi} \cdot \fSize(a, \epsilon) \cdot \abs{A}^{1+\epsilon}\),
  where \(a\) is the maximum arity of a relation symbol in \(\sigma_\phi\),
  and \(\sigma_\phi\) and \(a\) can be computed from \(\phi\) and \(\sigma\).

  It remains to consider the case where \(\phi\) is the result of applying rule~\eqref{def:cgfoc-P},
  that is, \(\phi\) is of the form
  \(\Land_{i=1}^k R_i(z_{i,1}, \dots, z_{i, \ar(R_i)}) \land \Pred(t_1, \dots, t_m)\).
  Let \(\tx'\) be a tuple of pairwise distinct variables such that
  \(\tilde{x}' = \bigcup_{i=1}^m \free(t_i)\).
We set \(\sigma_\phi \deff \sigma \uplus \set{R_\phi}\)
  for a fresh relation symbol \(R_\phi\) of arity \(\ell \deff \ar(R_\phi) \deff \abs{\tx'}\),
  and we set \(\phi'(\tx) \deff
  \Land_{i=1}^k R_i(z_{i,1}, \dots, z_{i, \ar(R_i)}) \land R_\phi(\tx')\).
  Furthermore, we let \(\A_\phi\) be the \(\sigma_\phi\)-expansion of \(\A\)
  where \(R_\phi(\A_\phi)\) is the set of all tuples \(\tw \in A^\ell\)
  such that the set \(\tilde{w}\) of entries is a clique in \(G_\A\)
  (repeated entries are allowed)
  and where
  \((\sem{t_1}^{(\A, \beta_{\tw})}, \dots, \sem{t_m}^{(\A, \beta_{\tw})}) \in \sem{\Pred}\)
  for \(\beta_{\tw}(x'_i) \deff w_i\) for all \(i \in [\ell]\).
  Since all tuples of \(R_\phi(\A_\phi)\) form a clique in \(G_\A\),
  it holds that \(G_\A = G_{\A_\phi}\).
  Moreover, for all \(\tv \in A^{\abs{\tx}}\),
  it holds that \(\A \models \phi(\tv)\) if and only if \(\A_\phi \models \phi'(\tv)\).

  There is an \(\FO[\sigma]\) formula \(\phi_{\textup{clique}}(\tx')\)
  that only depends on \(\sigma\) and \(\tx'\)
  such that, for all \(\tw \in A^{\smallabs{\tx'}}\),
  it holds that \(\A \models \phi_{\textup{clique}}(\tw)\)
  if and only if \(\tilde{w}\) is a clique in \(G_\A\).
  If \(\ell = 0\), there is only a single tuple \(\tw = () \in A^0\) to evaluate;
  for \(\ell \geq 1\), by \cref{lem:nowhere-dense-cliques},
  the number of such tuples \(\tw\) is bounded by
  \(\ell^\ell \cdot \fSize(\ell, \epsilon) \cdot \abs{A}^{1+\epsilon}\).
  Hence, by applying the enumeration result of \cref{thm:fo-testing-enumeration}
  to \(\A\) and \(\phi_{\textup{clique}}\),
  we can compute a list of all such tuples \(\tw\) in time
  \begin{align*}
    &\fFOTE(\phi_{\textup{clique}}, \sigma, \epsilon) \cdot \abs{A}^{1+\epsilon}
    + \fFOTE(\phi_{\textup{clique}}, \sigma, \epsilon) \cdot
    \ell^\ell \cdot \fSize(\ell, \epsilon) \cdot \abs{A}^{1+\epsilon}\\
    =
    \ &\fFOTE(\phi_{\textup{clique}}, \sigma, \epsilon)
    \cdot \bigl(1+ \ell^\ell \cdot \fSize(\ell, \epsilon)\bigr) \cdot \abs{A}^{1+\epsilon}.
  \end{align*}

  It remains to check for every such tuple \(\tw\) whether we have
  \((\sem{t_1}^{(\A, \beta_{\tw})}, \dots, \sem{t_m}^{(\A, \beta_{\tw})}) \in \sem{\Pred}\).
  For this, note that the counting terms \(t_1, \dots, t_m\) are built using
  rules~\eqref{def:foc-countingterm}--\eqref{def:foc-plustimesterm},
  that is, they are built using integer constants, addition, multiplication,
  and \#-terms of the form \(\FOCCount{\ty}{\psi}\).
  We recursively compute \(\sigma_\psi\), \(\psi'\), and \(\A_\psi\).
  By the induction hypothesis, these can be computed in time
  \(\fPre(\psi, \sigma, \epsilon) \cdot \abs{A}^{1+\epsilon}\).
  Then \(\sem{\FOCCount{\ty}{\psi}}^{(\A, \beta_{\tw})}
  = \sem{\FOCCount{\ty}{\psi'}}^{(\A_\psi, \beta_{\tw})}\),
  and \(\psi' \in \FO[\sigma_\psi]\).
  Hence, we can run the algorithm from \cref{thm:fo-counting} on \(\A_\psi\) and \(\psi'\).
  After preprocessing in time
  \(\fCountQA(\psi', \sigma_\psi, \epsilon) \cdot \abs{A}^{1+\epsilon}\),
  we can compute \(\sem{\FOCCount{\ty}{\psi'}}^{(\A_\psi, \beta_{\tw})}\)
  for a single tuple \(\tw\) in time \(\fCountQA(\psi', \sigma_\psi, \epsilon)\),
  so we can compute the results after the preprocessing in total time
  \(\fCountQA(\psi', \sigma_\psi, \epsilon) \cdot \ell^\ell
  \cdot \fSize(\ell, \epsilon) \cdot \abs{A}^{1+\epsilon}\)
  for all tuples \(\tw \in A^\ell\) such that \(\tilde{w}\)
  is a clique in \(G_\A\).

  Doing this for every outermost \#-term occurring in \(t_1, \dots, t_m\)
  (with any nested \#-terms already handled by the induction hypothesis)
  allows us to compute the values
  \(\sem{t_1}^{(\A, \beta_{\tw})}, \dots, \sem{t_m}^{(\A, \beta_{\tw})}\)
  for all tuples \(\tw \in A^\ell\) such that \(\tilde{w}\) is a clique in \(G_\A\)
  in time \(f'(\phi, \sigma, \epsilon) \cdot \abs{A}^{1+\epsilon}\)
  for a computable function \(f'\).
  Once we have computed \(\sem{t_1}^{(\A, \beta_{\tw})}, \dots, \sem{t_m}^{(\A, \beta_{\tw})}\),
  we can check whether \(\tw \in R_\phi(\A_\phi)\)
  with an oracle call to \(\sem{\Pred}\) in time \(\bigO(1)\).

  Thus, all in all we can define a computable function \(\fPre\)
  based on \(\fSize\), \(\fCountQA\), and \(\fFOTE\)
  such that we can compute \(\A_\phi\)
  in time \(\fPre(\phi, \sigma, \epsilon) \cdot \abs{A}^{1+\epsilon}\).
\end{proof}

Based on \cref{lem:cgfoc-preprocessing},
we can now prove \cref{thm:answering-enumeration},
which we repeat for convenience.

\answeringAndEnumeration*

\begin{proof}
  First, we prove \cref{thm:answering-enumeration}\ref{item:answering}.
  Let
  \(\fCountQA\),
  \(\fFOTE\),
  \(\fSize\),
  and \(\fPre\)
  be the functions from
  \cref{thm:fo-counting},
  \cref{thm:fo-testing-enumeration},
  \cref{lem:nowhere-dense-cliques},
  and \cref{lem:cgfoc-preprocessing},
  respectively.

  If \(\sigma \neq \sigma(\xi)\),
  then we can let \(\A'\) be the \(\sigma(\xi)\)-reduct of \(\A\)
  and continue with \(\A'\) instead of \(\A\).
  By \cref{lem:nowhere-dense-cliques},
  \(\norm{\A} \leq \abs{\sigma} \cdot \fSize(k, \epsilon) \cdot \abs{A}^{1+\epsilon}\),
  where \(k\) is the maximum arity of a relation symbol in \(\sigma\).
  Since we can compute \(\A'\) from \(\A\) in time linear in \(\norm{\A}\),
  in the following, we may assume that \(\sigma = \sigma(\xi)\).

  If \(\xi(\tx) = \phi(\tx)\) for a \(\cgFOC\) formula \(\phi\),
  then we apply \cref{lem:cgfoc-preprocessing}
  to \(\A\) and \(\phi\), and, in time
  \(\fPre(\phi, \sigma, \epsilon) \cdot \abs{A}^{1+\epsilon}\),
  we obtain a signature \(\sigma_\phi \supseteq \sigma\),
  an \(\FO[\sigma_\phi]\) formula \(\phi'\),
  and a \(\sigma_\phi\)-expansion \(\A_\phi\) of \(\A\) with \(G_\A = G_{\A_\phi}\)
  such that for all \(\tv \in A^{\abs{\tx}}\),
  we have \(\A \models \phi(\tv)\) if and only if \(\A_\phi \models \phi'(\tv)\).
Then, since \(G_{\A_\phi} \in \C\),
  we can apply \cref{thm:fo-testing-enumeration} to \(\A_\phi\) and \(\phi'\).
  Thus, after preprocessing in time
  \(\fFOTE(\phi', \sigma_\phi, \epsilon) \cdot \abs{A}^{1+\epsilon}\),
  given a tuple \(\tv \in A^{\abs{\tx}}\),
  we can output \(\sem{\phi'(\tv)}^{\A_\phi} = \sem{\phi(\tv)}^\A\)
  in time \(\fFOTE(\phi', \sigma_\phi, \epsilon)\).
  Since \(\sigma_\phi\) and \(\phi'\) only depend on \(\phi\) and \(\sigma\),
  we can set
  \(f_a(\phi, \sigma, \epsilon) \deff \fPre(\phi, \sigma, \epsilon)
  + \fFOTE(\phi', \sigma_\phi, \epsilon)\),
  which then satisfies the requirements of \cref{thm:answering-enumeration}\ref{item:answering}.

  If \(\xi(\tx)\) is a \(\cgFOC\) counting term,
  then \(\xi\) is built using rules~\eqref{def:foc-countingterm}--\eqref{def:foc-plustimesterm},
  that is, it is built using integer constants, addition, multiplication,
  and \#-terms of the form \(\FOCCount{\ty}{\psi(\tx', \ty)}\).
  For every such outermost \#-term \(\FOCCount{\ty}{\psi(\tx', \ty)}\) in \(\xi\),
  we apply \cref{lem:cgfoc-preprocessing} to \(\A\) and \(\psi\), and, in time
  \(\fPre(\psi, \sigma, \epsilon) \cdot \abs{A}^{1+\epsilon}\),
  we obtain a signature \(\sigma_\psi \supseteq \sigma\),
  an \(\FO[\sigma_\psi]\) formula \(\psi'\),
  and a \(\sigma_\psi\)-expansion \(\A_\psi\) of \(\A\) with \(G_\A = G_{\A_\psi}\)
  such that for all \(\tv' \in A^{\smallabs{\tx'}}\)
  and \(\tw \in A^{\abs{\ty}}\),
  we have \(\A \models \psi(\tv', \tw)\)
  if and only if \(\A_\psi \models \psi'(\tv', \tw)\).
  Hence, it holds that
  \(\sem{\FOCCount{\ty}{\psi(\tv', \ty)}}^\A = \sem{\FOCCount{\ty}{\psi'(\tv', \ty)}}^{\A_\psi}\).
  Since \(\psi'\) is a first-order formula and \(G_{\A_\psi} \in \C\),
  we can apply \cref{thm:fo-counting} to \(\A_\psi\) and \(\psi'\).
  Thus, after preprocessing in time
  \(\fCountQA(\psi', \sigma_\psi, \epsilon) \cdot \abs{A}^{1+\epsilon}\),
  given a tuple \(\tv' \in A^{\smallabs{\tx'}}\),
  we can compute
  \(\sem{\FOCCount{\ty}{\psi(\tv', \ty)}}^\A =
  \sem{\FOCCount{\ty}{\psi'(\tv', \ty)}}^{\A_\psi} =
  \bigabs{\psi'(\A_\psi, \tv')}\)
  in time \(\fCountQA(\psi', \sigma_\psi, \epsilon)\).
  Hence, whenever we are given a tuple \(\tv \in A^{\abs{\tx}}\),
  we can efficiently compute the results of the outermost \#-terms in \(\xi\),
  and then compute \(\sem{\xi(\tv)}^\A\) with additions and multiplications.
  Thus, all in all, we can set \(f_a(\xi, \sigma, \epsilon)\) such that the preprocessing
  takes time at most \(f_a(\xi, \sigma, \epsilon) \cdot \abs{A}^{1+\epsilon}\)
  and, given \(\tv \in A^{\abs{\tx}}\), the computation of \(\sem{\xi(\tv)}^\A\)
  takes time at most \(f_a(\xi, \sigma, \epsilon)\).

  Now, we prove \cref{thm:answering-enumeration}\ref{item:enumeration}.
  As above, we may assume that \(\sigma = \sigma(\xi)\).
We apply \cref{lem:cgfoc-preprocessing} to \(\A\) and \(\xi\), and, in time
  \(\fPre(\xi, \sigma, \epsilon) \cdot \abs{A}^{1+\epsilon}\),
  we obtain a signature \(\sigma_\xi \supseteq \sigma\),
  an \(\FO[\sigma_\xi]\) formula \(\xi'\),
  and a \(\sigma_\xi\)-expansion \(\A_\xi\) of \(\A\) with \(G_\A = G_{\A_\xi}\)
  such that for all \(\tv \in A^{\abs{\tx}}\),
  we have \(\A \models \xi(\tv)\) if and only if \(\A_\xi \models \xi'(\tv)\).
Then, since \(G_{\A_\xi} \in \C\),
  we can apply \cref{thm:fo-testing-enumeration} to \(\A_\xi\) and \(\xi'\).
  Thus, after preprocessing in time
  \(\fFOTE(\xi', \sigma_\xi, \epsilon) \cdot \abs{A}^{1+\epsilon}\),
  we can enumerate all tuples \(\tv \in A^{\abs{\tx}}\) such that \(\A \models \xi(\tv)\)
  with \(\fFOTE(\xi', \sigma_\xi, \epsilon)\) delay in lexicographic order,
  without duplicates.
Since \(\sigma_\xi\) and \(\xi'\) only depend on \(\xi\) and \(\sigma\),
  we can set
  \(f_b(\xi, \sigma, \epsilon) \deff \fPre(\xi, \sigma, \epsilon)
  + \fFOTE(\xi', \sigma_\xi, \epsilon)\),
  which then satisfies the requirements of \cref{thm:answering-enumeration}\ref{item:enumeration}.

  Finally, it suffices to set \(\fcgFOC(\xi, \sigma, \epsilon)
  \deff \max\bigl(f_a(\xi, \sigma, \epsilon), f_b(\xi, \sigma, \epsilon)\bigr)\).
\end{proof}
 \section{Proof of Lemma~\ref{lem:game-tree}}
\label{sec:game-tree-proof}

\gameTree*

\begin{proof}
  Let \(\lambda, t\) be the computable functions from \cref{lem:winning-strategy},
  let \(\fNC\) be the computable function from \cref{thm:neighbourhood-cover},
  and let \(\fSize\) be the computable function from \cref{lem:nowhere-dense-cliques} for \(\C\).

  Let the input of the algorithm be \(r \in \N\), a graph \(G \in \C\),
  and \(\epsilon, \epsilon' \in \Qpos\)
  such that \((1+\epsilon')^{t(2r)+1} \leq 1 + \epsilon\).

  We briefly repeat the construction described in \cref{sec:counting}.
  For the root \(()\) of our tree \(T\),
  we set \(G_{()} \deff G\),
  and we let \(\NC_{()}\) be a distance-\(r\) neighbourhood cover of \(G\)
  of radius at most \(2r\).
  Furthermore, we set \(\pi_{()} \deff ()\) as the initial partial play
  in the generalised \((\lambda(2r), 2r)\)-splitter game on \(G\).
To every inner node \(\tCentres = (\centre_1, \dots, \centre_i)\) of \(T\)
  for some \(i \in \N\),
  we add a child \(\tCentres \centre = (\centre_1, \dots, \centre_i, \centre)\)
  for every centre \(\centre \in \centres_{\tCentres}\),
  and we set \(v_{\tCentres \centre} \deff \centre\),
  \(H_{\tCentres \centre} \deff G_{\tCentres}[\centreToCluster_{\tCentres}(\centre)]\),
  and \(W_{\tCentres \centre}\) is Splitter's answer according to
  the winning strategy from \cref{lem:winning-strategy}
  for the generalised \((\lambda(2r), 2r)\)-splitter game on \(G\)
  when Connector continues the partial play \(\pi_{\tCentres}\)
  by choosing \(H_{\tCentres \centre}\) and \(v_{\tCentres \centre}\).
  We set \(\pi_{\tCentres \centre}\) to the resulting partial play after these choices.
  If \(W_{\tCentres \centre} = V(H_{\tCentres \centre})\),
  then
  \(\neighb{2r}{H_{\tCentres \centre}}{v_{\tCentres \centre}}
  \setminus W_{\tCentres \centre} = \emptyset\),
  so Splitter wins.
  In this case, \(\tCentres \centre\) is a leaf of the tree.
Otherwise, we have \(W_{\tCentres \centre} \subsetneq V(H_{\tCentres \centre})\).
  Since
  \(V(H_{\tCentres \centre})
  = \centreToCluster_{\tCentres}(\centre)
  \overset{\eqref{eq:nc-radius}}{=}
  \neighb{2r}{G_{\tCentres}[\centreToCluster_{\tCentres}(\centre)]}{\centre}
  = \neighb{2r}{H_{\tCentres \centre}}{\centre}\),
  this implies that
  \(\neighb{2r}{H_{\tCentres \centre}}{v_{\tCentres \centre}}
  \setminus W_{\tCentres \centre} \neq \emptyset\),
  so the play continues, and \(\tCentres \centre\) is not a leaf.
  We let
  \(G_{\tCentres \centre} \deff
  \Neighb{2r}{H_{\tCentres \centre}}{v_{\tCentres \centre}} \setminus W_{\tCentres \centre}
  = H_{\tCentres \centre} \setminus W_{\tCentres \centre}\).
  Furthermore, we let \(\NC_{\tCentres \centre}\) be a distance-\(r\) neighbourhood cover
  of \(G_{\tCentres \centre}\) of radius at most \(2r\).

  All graphs constructed in the game tree are (induced) subgraphs of \(G\),
  and \(\C\) is assumed to be closed under taking subgraphs.
Thus, for the construction of the distance-\(r\) neighbourhood cover
  \(\NC_{\tCentres}\) of a graph \(G_{\tCentres}\),
  we can use \cref{thm:neighbourhood-cover} on input \(\epsilon', r, G_{\tCentres}\).
  This guarantees that \(\NC_{\tCentres}\) has radius at most \(2r\)
  and overlap at most \(\fNC(r, \epsilon') \cdot \abs{V(G_{\tCentres})}^{\epsilon'}\).
  Hence, by \cref{eq:nc-overlap},
  \(\sum_{\centre \in \centres_{\tCentres}}
  \abs{\centreToCluster_{\tCentres}(\centre)}
  \leq \fNC(r, \epsilon') \cdot \abs{V(G_{\tCentres})}^{1+\epsilon'}\).
Furthermore, for the computation of Splitter's answer,
  we use \cref{lem:winning-strategy}.
  This guarantees that Splitter wins
  the generalised \((\lambda(2r), 2r)\)-splitter game
  in at most \(t(2r)\) rounds, and \(T\) has height at most \(t(2r)\).

  It remains to show that the computation runs in time
  \(\fGameTree(r, \epsilon, \epsilon') \cdot \abs{V(G)}^{1+\epsilon}\)
  for a suitable function \(\fGameTree\).
  The computation of \(\NC_{()}\) takes time
  \(\fNC(r, \epsilon') \cdot \abs{V(G)}^{1+\epsilon'}\).
  Moreover, for every inner node \(\tCentres\) of \(T\)
  and for every centre \(\centre \in \centres_{\tCentres}\),
  the computation of Splitter's answer \(W_{\tCentres \centre}\) takes time
  \(\bigO\bigl(\abs{V(H_{\tCentres \centre})} + \abs{E(H_{\tCentres \centre})}
  + 2r \cdot t(2r)\bigr)
  \leq \bigO\bigl(\fSize(2,\epsilon') \cdot
  \abs{\centreToCluster_{\tCentres}(\centre)}^{1+\epsilon'} + 2r \cdot t(2r)\bigr)\),
  where the inequality holds by \cref{lem:nowhere-dense-cliques}
  and because \(H_{\tCentres \centre} \in \C\).
  Note that the setting considered here slightly differs
  from the setting in \cref{lem:winning-strategy},
  since \cref{lem:winning-strategy} updates the data structures
  during a single play, while we do the computations along several branches
  of possible plays.
  However, by doing the computations in a \emph{depth-first} fashion,
  so first finishing a single play and then going back one round
  and trying a different choice for Connector,
  and by remembering the changes made while updating the data structures
  and reverting the changes when moving to another branch,
  we can still use \cref{lem:winning-strategy} with the same running-time bounds.

  If \(\tCentres \centre\) is not a leaf,
  then we compute the distance-\(r\) neighbourhood cover \(\NC_{\tCentres \centre}\).
  This takes time
  \(\fNC(r, \epsilon') \cdot \abs{V(G_{\tCentres \centre})}^{1+\epsilon'} \leq
  \fNC(r, \epsilon') \cdot \abs{\centreToCluster_{\tCentres}(\centre)}^{1+\epsilon'}\).
  Thus, there is a computable function \(f'\) such that the computation
  of all children at an inner node \(\tCentres\) takes time at most
  \begin{align*}
    \sum_{\centre \in \centres_{\tCentres}}
    f'(r, \epsilon') \cdot \abs{\centreToCluster_{\tCentres}(\centre)}^{1+\epsilon'}
    &\leq f'(r, \epsilon') \cdot \Bigl(\sum_{\centre \in \centres_{\tCentres}}
    \abs{\centreToCluster_{\tCentres}(\centre)}\Bigr)^{1+\epsilon'}\\
    &\leq f'(r, \epsilon') \cdot \Bigl(\fNC(r, \epsilon')
    \cdot \abs{V(G_{\tCentres})}^{1+\epsilon'}\Bigr)^{1+\epsilon'}\\
    &= f'(r, \epsilon') \cdot \fNC(r, \epsilon')^{1+\epsilon'}
    \cdot \abs{V(G_{\tCentres})}^{(1+\epsilon')^2}.
  \end{align*}
  Here, the first inequality uses that the \(p\)-norm is bounded by the \(1\)-norm,
  \ie \(\bigl(\sum_{i=1}^m x_i^p\bigr)^{1/p} \leq \sum_{i=1}^m x_i\)
  for all \(p \geq 1\),
  \(m \in \N\),
  and \(x_1, \dots, x_m \geq 0\),
  which implies \(\sum_{i=1}^m x_i^p \leq \bigl(\sum_{i=1}^m x_i\bigr)^p\).

  Let \(I(T)\) be the set of inner nodes of \(T\).
  By induction on \(i \in \N\), we show that
  \[\sum_{\substack{\tCentres \in I(T),\\ \smallabs{\tCentres} = i}} \abs{V(G_{\tCentres})}
  \leq \fNC(r, \epsilon')^{\bigl((1+\epsilon')^i-1\bigr)/\epsilon'}
  \cdot \abs{V(G)}^{(1+\epsilon')^i},\]
  which bounds the combined size of all graphs \(G_{\tCentres}\)
  for inner nodes \(\tCentres\) at depth \(i\) in \(T\).
  Indeed, for \(i = 0\), the statement trivially holds.
  For the step \(i \to i+1\), we have
  \begin{align*}
    \sum_{\substack{\tCentres \in I(T),\\ \smallabs{\tCentres} = i+1}} \abs{V(G_{\tCentres})}
    &\leq \sum_{\substack{\tCentres \in I(T),\\ \smallabs{\tCentres} = i}}
    \sum_{\centre \in \centres_{\tCentres}} \abs{\centreToCluster_{\tCentres}(\centre)}\\
    &\leq \sum_{\substack{\tCentres \in I(T),\\ \smallabs{\tCentres} = i}}
    \fNC(r, \epsilon') \cdot \abs{V(G_{\tCentres})}^{1+\epsilon'}\\
    &\leq \fNC(r, \epsilon') \cdot
    \Bigl(\sum_{\substack{\tCentres \in I(T),\\ \smallabs{\tCentres} = i}}
    \abs{V(G_{\tCentres})}\Bigr)^{1+\epsilon'}\\
    &\leq \fNC(r, \epsilon') \cdot
    \Bigl(\fNC(r, \epsilon')^{\bigl((1+\epsilon')^i-1\bigr)/\epsilon'}
    \cdot \abs{V(G)}^{(1+\epsilon')^i}\Bigr)^{1+\epsilon'}\\
    &= \fNC(r, \epsilon')^{\bigl((1+\epsilon')^{i+1}-1\bigr)/\epsilon'}
    \cdot \abs{V(G)}^{(1+\epsilon')^{i+1}},
  \end{align*}
  which is the statement for \(i+1\) instead of \(i\).

  We combine the bound on the running time for computing all children
  of an inner node \(\tCentres\)
  with the bound on the combined size of all graphs at inner nodes at depth \(i\).
  Furthermore, we use the fact that the tree \(T\) has height at most \(t(2r)\)
  and that the whole tree is computed by running the above-mentioned steps
  for all inner nodes (up to depth \(t(2r)-1\)),
  which also creates all leaves.
  This yields that the total running time for computing all nodes
  apart from the root can be bounded by
  \begin{align*}
    &~
    \sum_{i=0}^{t(2r)-1}
    \sum_{\substack{\tCentres \in I(T),\\ \smallabs{\tCentres} = i}}
    f'(r, \epsilon') \cdot \fNC(r, \epsilon')^{1+\epsilon'}
    \cdot \abs{V(G_{\tCentres})}^{(1+\epsilon')^2}\\
    \leq
    &~
    f'(r, \epsilon') \cdot \fNC(r, \epsilon')^{1+\epsilon'}
    \sum_{i=0}^{t(2r)-1}
    \Bigl(
      \sum_{\substack{\tCentres \in I(T),\\ \smallabs{\tCentres} = i}}
      \abs{V(G_{\tCentres})}
    \Bigr)^{(1+\epsilon')^2}\\
    \leq
    &~
    f'(r, \epsilon') \cdot \fNC(r, \epsilon')^{1+\epsilon'}
    \sum_{i=0}^{t(2r)-1}
    \Bigl(
      \fNC(r, \epsilon')^{\bigl((1+\epsilon')^i-1\bigr)/\epsilon'}
      \cdot \abs{V(G)}^{(1+\epsilon')^i}
    \Bigr)^{(1+\epsilon')^2}\\
    \leq
    &~
    f'(r, \epsilon') \cdot \fNC(r, \epsilon')^{1+\epsilon'}
    \cdot t(2r) \cdot
    \fNC(r, \epsilon')^{(1+\epsilon')^2 \cdot \bigl((1+\epsilon')^{t(2r)-1}-1\bigr)/\epsilon'}
    \cdot \abs{V(G)}^{(1+\epsilon')^{t(2r)+1}}.
  \end{align*}
  All in all, combined with the fact that
  \((1+\epsilon')^{t(2r)+1} \leq 1+\epsilon\)
  and that the neighbourhood cover for the root can be computed in time
  \(\fNC(r, \epsilon') \cdot \abs{V(G)}^{1+\epsilon'}\),
  this shows that there is a computable function \(\fGameTree\)
  such that the overall running time can be bounded by
  \(\fGameTree(r, \epsilon, \epsilon') \cdot \abs{V(G)}^{1+\epsilon}\).
\end{proof}
 \section{Proof of Lemmas~\ref{lem:reduction-connected} and~\ref{lem:delta-connected-local}
and of Corollary~\ref{cor:removal}}
\label{sec:tools-formulae-proofs}

\reductionConnected*
\begin{proof}
  We prove the result by induction on the number of connected components of \(G\).
  If \(G\) is connected, then we choose \(m \deff 1\), \(G_1 \deff G\),
  and \(P(X_1) \deff X_1\),
  which implies that
  \(\phi_{G,1}(\tx_{\relabel{G}(V_1)}) = \phi_{G,1}(\tx) = \phi_G(\tx)\).

  Now suppose \(G\) is not connected.
  Let \(V_1 \in \components(G)\) be the connected component of \(G\)
  containing the minimal vertex \(\min\bigl(V(G)\bigr) = \relabel{G}^{-1}(1)\),
  and let \(\tilde{V} \deff V(G) \setminus V_1\).
  We let \(G_1 \deff G[V_1]\) and \(\tilde{G} \deff G[\tilde{V}]\).
  Then \(G_1\) is connected,
  and \(\tilde{G}\) has \(\abs{\components(G)}-1\) connected components.
  We apply the induction hypothesis to \(\tilde{G}\)
  and compute \(\tilde{m} \in \N\),
  connected graphs \(\tilde{G}_1, \dots, \tilde{G}_{\tilde{m}}\)
  such that \(V(\tilde{G}_i)\) is a union of connected components of \(\tilde{G}\),
  and a polynomial \(\tilde{P}(X_1, \dots, X_{\tilde{m}})\).
  Since \(\tilde{G}\) is the graph \(G\) without the connected component \(V_1\),
  the \(V(\tilde{G}_i)\) are also unions of connected components of \(G\).

  Let \(\G_{\neg G}\) be the set of graphs \(H\) with
  \(V(H) = V(G)\),
  \(H[V_1] = G[V_1]\), \(H[\tilde{V}] = G[\tilde{V}]\), and \(H \neq G\).
  Since such a graph \(H\) coincides with \(G\) on \(V_1\) and \(\tilde{V}\),
  and there are no edges between \(V_1\) and \(\tilde{V}\) in \(G\),
  the graph \(H\) needs to have an edge between \(V_1\) and \(\tilde{V}\),
  connecting two connected components of \(G\),
  so \(H\) has fewer connected components than \(G\).
  Meanwhile, every connected component of \(H\) is a union of connected components of \(G\).
  Let \(g \deff \abs{\G_{\neg G}}\),
  and let \(H_1, \dots, H_g\) be such that \(\G_{\neg G} = \set{H_1, \dots, H_g}\).
  For every \(i \in [g]\), we apply the induction hypothesis to \(H_i\)
  and compute \(m_i \in \N\),
  connected graphs \(H_{i,1}, \dots, H_{i, m_i}\)
  with \(V(H_{i,j})\) being a union of connected components of \(H_i\),
  and a polynomial \(P_i(X_1, \dots, X_{m_i})\).
  Then the \(V(H_{i,j})\) are also unions of connected components of \(G\).

  Set \(m \deff 1 + \tilde{m} + \sum_{i=1}^g m_i\),
  and let \(G_{i+1} \deff \tilde{G}_i\) and \(V_{i+1} \deff V(\tilde{G}_i)\)
  for all \(i \in [\tilde{m}]\).
  We let \(\mu \colon \setc{(i,j)}{i \in [g], j \in [m_i]} \to \N\)
  be the function with \(\mu(i,j) \deff 1 + \tilde{m} + j + \sum_{s=1}^{i-1} m_s\)
  for all \(i \in [g]\) and \(j \in [m_i]\).
  This function can be used to enumerate the graphs \(H_{i,j}\).
  That is, we set \(G_{\mu(i,j)} \deff H_{i,j}\)
  and \(V_{\mu(i,j)} \deff V(H_{i,j})\)
  for all \(i \in [g]\) and \(j \in [m_i]\).
  Furthermore, we set
  \[P(X_1, \dots, X_m) \deff X_1 \cdot \tilde{P}(X_2, \dots, X_{\tilde{m} + 1})
  - \sum_{i \in [g]} P_i(X_{\mu(i,1)}, \dots, X_{\mu(i,m_i)}).\]
Let
  \begin{align*}
    \phi_{G,i}(\tx_{\relabel{G}(V_i)})
    &\deff \delta_{G_i, r}(\tx_{\relabel{G}(V_i)}) \land
    \Land_{\substack{I \in \components(G),\\ I \subseteq V_i}}
    \psi_{G,I}(\tx_{\relabel{G}(I)})
    \quad \text{for all } i \in [m],\\
    \phi_{\tilde{G}}(\tx_{\relabel{G}(\tilde{V})})
    &\deff \delta_{\tilde{G},r}(\tx_{\relabel{G}(\tilde{V})})
    \land \Land_{I \in \components(\tilde{G})}
    \psi_{G,I}(\tx_{\relabel{G}(I)}),\\
    \shortintertext{and, for \(i \in [g]\), let}
    \phi_{H_i}(\tx)
    &\deff \delta_{H_i,r}(\tx)
    \land
    \Land_{I \in \components(H_i)}
    \psi_{H_i, I}(\tx_{\relabel{G}(I)})\\
    \shortintertext{with}
    \psi_{H_i, I}(\tx_{\relabel{G}(I)})
    &\deff
    \Land_{\substack{I' \in \components(G),\\ I' \subseteq I}}
    \psi_{G,I'}(\tx_{\relabel{G}(I')}).
  \end{align*}
  For every \(\sigma\)-structure \(\A\) and all \(\tw \in A^k\),
  by the definitions of the formulae above,
  it holds that
  \(\A \models \phi_{G,1}(\tw_{\relabel{G}(V_1)})\) and
  \(\A \models \phi_{\tilde{G}}(\tw_{\relabel{G}(\tilde{V})})\)
  if and only if there is a graph \(H \in \G_{\neg G} \cup \set{G}\)
  with \(\A \models \phi_H(\tw)\).
  Furthermore, for every \(\tw \in A^k\),
  there is at most one graph \(H \in \G_{\neg G} \cup \set{G}\)
  with \(\A \models \phi_H(\tw)\).
For all \(\ell \leq k\) and \(\tv \in A^\ell\),
  this implies that
  \[
    \bigabs{\phi_G(\A, \tv)}
    = \bigabs{\phi_{G,1}(\A, \tv_{\relabel{G}(V_1)})}
    \cdot \bigabs{\phi_{\tilde{G}}(\A, \tv_{\relabel{G}(\tilde{V})})}
    - \sum_{i \in [g]} \bigabs{\phi_{H_i}(\A, \tv)}.
  \]
  Moreover, by the induction hypothesis, it holds that
  \[
    \bigabs{\phi_{\tilde{G}}(\A, \tv_{\relabel{G}(\tilde{V})})}
    = \tilde{P}\bigl(
      \bigabs{\phi_{G,2}(\A, \tv_{\relabel{G}(V_2)})},
      \dots,
      \bigabs{\phi_{G,\tilde{m}+1}(\A, \tv_{\relabel{G}(V_{\tilde{m}+1})})}
    \bigr),
  \]
  where we use
  \(\bigl(\tv_{\relabel{G}(V(\tilde{G}))}\bigr)_{\relabel{\tilde{G}}(W)}
  = \tv_{\relabel{G}(W)}\)
  for all \(W \subseteq V(\tilde{G})\).
  Further, again by the induction hypothesis, we have
  \[
    \abs{\phi_{H_i}(\A, \tv)}
    = P_i\bigl(
      \bigabs{\phi_{G,\mu(i,1)}(\A, \tv_{\relabel{G}(V_{\mu(i,1)})})},
      \dots,
      \bigabs{\phi_{G,\mu(i,m_i)}(\A, \tv_{\relabel{G}(V_{\mu(i,m_i)})})}
    \bigr),
  \]
  where we use \(\relabel{G} = \relabel{H_i}\),
  for all \(i \in [g]\).
  All in all, this shows that
  \[
    \abs{\phi_G(\A, \tv)}
    = P\bigl(
      \abs{\phi_{G,1}(\A, \tv_{\relabel{G}(V_1)})},
      \dots,
      \abs{\phi_{G,m}(\A, \tv_{\relabel{G}(V_m)})}
    \bigr),
  \]
  and we have already seen that \(V_i\) is a union of connected components
  of \(G\) for all \(i \in [m]\).
  This concludes the proof.
\end{proof}

\deltaConnectedLocal*
\begin{proof}
  For both directions, we use that \(\nrA{\tv} \subseteq S\)
  implies \(\NrA{\tv} = \Neighbr{\A[S]}{\tv}\).

  For the forward direction,
  suppose we have \(\A \models \phi(\tv)\).
  Since \(G\) is connected,
  having \(\A \models \delta_{G,r}(\tv)\) implies that
  \(\dist^\A(v_1, v_i) \leq (k-1) \cdot r\) for all \(i \in [k]\),
  so \(\tilde{v} \subseteq \neighbA{(k-1)r}{v_1}\)
  and \(\nrA{\tv} \subseteq \neighbA{kr}{v_1} \subseteq S\).
  In particular, \(\tv \in S^k\),
  and, by the observation above,
  \(\NrA{\tv} = \Neighbr{\A[S]}{\tv}\).
  Since \(\phi(\tx)\) is \(r\)-local,
  we obtain
  \(\A \models \phi(\tv)
  \iff \NrA{\tv} \models \phi(\tv)
  \iff \Neighbr{\A[S]}{\tv} \models \phi(\tv)
  \iff \A[S] \models \phi(\tv)\).

  For the backward direction,
  suppose we have \(\tv \in S^k\) and \(\A[S] \models \phi(\tv)\).
  Again, \(\A[S] \models \delta_{G,r}(\tv)\) implies that
  \(\dist^{\A[S]}(v_1, v_i) \leq (k-1) \cdot r\) for all \(i \in [k]\),
  which implies that
  \(\dist^{\A}(v_1, v_i) \leq (k-1) \cdot r\) for all \(i \in [k]\).
  Thus, \(\nrA{\tv} \subseteq \neighbA{kr}{v_1} \subseteq S\)
  and \(\NrA{\tv} = \Neighbr{\A[S]}{\tv}\).
  By \(r\)-locality of \(\phi(\tx)\), again,
  we have
  \(\A[S] \models \phi(\tv)
  \iff \NrA{\tv} \models \phi(\tv)
  \iff \A \models \phi(\tv)\).
\end{proof}

\removalIterated*

\begin{proof}
  Let \(k, p, q \in \N\) with \(k + p \leq q\),
  let \(r \deff \radius(p,q)\),
  and let \(\sigma\) be a signature.
  We prove the result by induction on \(s \in \N\).

  For the base case \(s = 0\),
  there is only a single function \(\mu \colon [k] \to \set{0}\).
  Furthermore, we have \(\tz_\mu = \tz\),
  and, for every \(\sigma\)-structure \(\A\),
  we have \(\A \removaldistr d_1 \cdots d_s = \A\),
  and \(\tv_\mu = \tv\) for every \(\tv \in A^k\).
  Thus, the base case is trivially settled by setting
  \(\phi_\mu(\tz_\mu) \deff \phi(\tz)\).

  For the induction step from \(s\) to \(s+1\),
  let \(\tz = (z_1, \dots, z_k)\) be a tuple of pairwise distinct variables,
  let \(\phi(\tz)\) be an \(\FOplusSigmaPQ\) formula,
  and let \(\mu \colon [k] \to [0,s+1]\).
  We set \(\mu' \colon [k] \to [0,s]\) with \(\mu'(i) = \mu(i)\) if \(\mu(i) \neq s+1\)
  and \(\mu'(i) = 0\) if \(\mu(i) = s+1\).
  Let \(\phi_{\mu'}(\tz_{\mu'})\) be the
  \(\FOplusParams{\removalsigrs{\sigma}}{p}{q}\) formula
  obtained from applying the induction hypothesis to \(\phi(\tz)\) and \(\mu'\).
  Then, for all \(\sigma\)-structures \(\A\) with \(\abs{A} > s\),
  all distinct \(d_1, \dots, d_s \in A\),
  and all \(\tv = (v_1, \dots, v_k) \in A^k\) with
  \(v_i \not\in \set{d_1, \dots, d_s}\) if \(\mu'(i) = 0\)
  and \(v_i = d_{\mu'(i)}\) else,
  we have
  \begin{equation}
    \label{eq:removal-cor-0-to-s}
    \A \models \phi(\tv)
    \iff \A \removaldistr d_1 \cdots d_s \models \phi_{\mu'}(\tv_{\mu'}).
  \end{equation}

  Let \(k' \deff \abs{\setc{i \in [k]}{\mu'(i) = 0}}\),
  let \(i_1, \dots, i_{k'} \in [k]\) such that
  \(\tz_{\mu'} = (z_{i_1}, \dots, z_{i_{k'}})\),
  let \(\tz' \deff \tz_{\mu'}\),
  and let \(\phi'(\tz') \deff \phi_{\mu'}(\tz_{\mu'})\).
  Set \(I \deff \setc{j \in [k']}{\mu(i_j) = 0}\),
  and note that \([k'] \setminus I = \setc{j \in [k']}{\mu(i_j) = s+1}\)
  and \(\tz'_I = \tz_\mu\).
  Let \(\phi'_I(\tz'_I)\) be the
  \(\FOplusParams{\removalsig{\sigma}{r,s+1}}{p}{q}\) formula
  obtained from applying \cref{lem:removal}
  to \(\phi'(\tz')\) and \(I\).
  Then, for all \(\removalsigrs{\sigma}\)-structures \(\A\) with \(\abs{A} \geq 2\),
  for all \(d_{s+1} \in A\),
  and all \(\tv' = (v'_1, \dots, v'_{k'}) \in A^{k'}\)
  such that \(I = \setc{j \in [k']}{v'_j \neq d_{s+1}}\),
  we have
  \(\A \models \phi'(\tv') \iff \A \removaldistr d_{s+1} \models \phi'_I(\tv'_I)\).
  In particular,
  for all \(\sigma\)-structures \(\A\) with \(\abs{A} > s+1\),
  all distinct \(d_1, \dots, d_{s+1} \in A\),
  and all
  \(\tv' = (v'_1, \dots, v'_{k'}) \in (A \setminus \set{d_1, \dots, d_s})^{k'}\)
  such that \(I = \setc{j \in [k']}{v'_j \neq d_{s+1}}\),
  we have
  \begin{equation}
    \label{eq:removal-cor-s-to-s-plus-1}
    \A \removaldistr d_1 \cdots d_s \models \phi'(\tv')
    \iff \A \removaldistr d_1 \cdots d_s d_{s+1} \models \phi'_I(\tv'_I).
  \end{equation}

  We let \(\phi_\mu(\tz_\mu) \deff \phi'_I(\tz'_I)\).
  Then, for all \(\sigma\)-structures \(\A\) with \(\abs{A} > s+1\),
  all distinct \(d_1, \dots, d_{s+1} \in A\),
  and all \(\tv = (v_1, \dots, v_k) \in A^k\) with
  \(v_i \not\in \set{d_1, \dots, d_{s+1}}\) if \(\mu(i) = 0\)
  and \(v_i = d_{\mu(i)}\) else,
  we have \(\tv_{\mu'} \in (A \setminus \set{d_1, \dots, d_s})^{k'}\)
  and \((\tv_{\mu'})_I = \tv_\mu\).
  Furthermore,
  we have \(v_i \not\in \set{d_1, \dots, d_s}\) if \(\mu'(i) = 0\)
  and \(v_i = d_{\mu'(i)}\) else.
  Hence, by combining \cref{eq:removal-cor-0-to-s,eq:removal-cor-s-to-s-plus-1},
  we obtain that
  \begin{align*}
    \A \models \phi(\tv)
    &\overset{\eqref{eq:removal-cor-0-to-s}}{\iff}
    \A \removaldistr d_1 \cdots d_s \models \phi_{\mu'}(\tv_{\mu'})\\
    &\iff
    \A \removaldistr d_1 \cdots d_s \models \phi'(\tv_{\mu'})\\
    &\overset{\eqref{eq:removal-cor-s-to-s-plus-1}}{\iff}
    \A \removaldistr d_1 \cdots d_{s+1} \models \phi_\mu(\tv_\mu).
  \end{align*}
  This shows that \(\phi_\mu(\tz_\mu)\) satisfies the requirements of \cref{cor:removal},
  which completes the induction step.
\end{proof}
 \section{Conclusion}
\label{sec:conclusion}

In this paper, for every effectively nowhere dense graph class \(\C\),
we have presented an algorithm that can compute
\(\bigabs{\phi(\A, \tv)}
= \abs{\bigsetc{\tw \in A^{\abs{\ty}}}{\A \models \phi(\tv, \tw)}}\),
\ie the number of \(\tw \in A^{\abs{\ty}}\) with \(\A \models \phi(\tv, \tw)\),
in constant time for any given tuple \(\tv \in A^{\abs{\tx}}\)
after performing an almost-linear-time preprocessing step
on a given structure \(\A\) with \(G_\A \in \C\)
and a given first-order formula \(\phi(\tx, \ty)\)
(\cref{thm:fo-counting}).
Based on this algorithm,
we have presented algorithms for constant-time query answering
and constant-delay enumeration after almost-linear-time preprocessing
for the clique-guarded first-order logic with counting \(\cgFOC\)
on effectively nowhere dense classes
(\cref{thm:answering-enumeration}).

For the proof of \cref{thm:fo-counting},
we have introduced the concept of \emph{game trees},
and we have proved in \cref{lem:game-tree} that these can be computed efficiently
for effectively nowhere dense classes.
In addition, with \cref{lem:reduction-connected},
we have provided a tool for turning interdependent computations in various regions of a structure
into independent computations in connected regions.
We believe that these concepts and results can be of independent interest.

As discussed in the introduction,
the results of \cref{thm:answering-enumeration} are optimal:
assuming \(\FPT \neq \AWstar\),
they cannot be generalised to subgraph-closed classes
beyond nowhere dense classes,
slight generalisations of \(\cgFOC\)
are intractable on very simple sparse classes,
and even very restricted fragments of \(\cgFOC\)
are already intractable on very simple dense classes.
Meanwhile, it remains an interesting open question whether \cref{thm:fo-counting}
can be generalised to well-behaved dense classes
such as structurally nowhere dense classes,
monadically stable classes,
or classes of bounded twin-width or bounded merge-width.
Building on this, we believe that \(\FOC_1\) is a natural candidate
for a tractable counting logic on such classes.

\bibliography{main}

\appendix
\crefalias{section}{appendix}
\crefalias{subsection}{appendix}
\section*{Appendix}
\section{Details on Theorem~\ref{thm:neighbourhood-cover}}
\label{sec:details-neighbourhood-cover}

\neighbourhoodCover*

\cite[Theorem~6.2]{GroheKreutzerSiebertz_2017_NowhereDense} states the result
for graphs of size \(\abs{V(G)} \geq f(r, \epsilon)\)
for a computable function \(f\)
in case \(\C\) is effectively nowhere dense.
However, for graphs of size \(\abs{V(G)} < f(r, \epsilon)\),
we can use the trivial neighbourhood cover \(\NCtuple{}\)
with \(\centres = V(G)\), \(\vertexToCentre(v) = v\) for all \(v \in V(G)\),
and \(\centreToCluster(\centre) = \neighbr{G}{\centre}\)
for all \(\centre \in \centres = V(G)\),
which has overlap less than \(f(r, \epsilon)\),
and the running time for computing such a neighbourhood cover
for graphs of size less than \(f(r, \epsilon)\)
is bounded by a function of \(r\) and \(\epsilon\)
(computable if \(\C\) is effectively nowhere dense).
Hence, \(\fNC(r, \epsilon)\) can be chosen such that the result
also holds for graphs of size less than \(\fNC(r, \epsilon)\).

Before discussing the data structure in detail,
we first introduce some notation used in~\cite{GroheKreutzerSiebertz_2017_NowhereDense}.
For a graph \(G\), a linear order \(<\) on \(V(G)\),
and for \(k \in \N\) and \(s, t \in V(G)\),
we say that \(t\) is \emph{weakly \(k\)-accessible} from \(s\) in \(G\) with respect to \(<\)
if \(t < s\) and there is a path \(P\) of length at most \(k\) from \(s\) to \(t\) in \(G\)
such that for all \(u \in V(P)\), we have \(t \leq u\).
Let \(\WReach_k(G, <, s)\) be the set of all vertices in \(V(G)\) that are weakly \(k\)-accessible
from \(s\) in \(G\) with respect to \(<\),
and let \(\WReach_k[G, <, s] \deff \WReach_k(G, <, s) \cup \set{s}\).

For computing sparse neighbourhood covers,
the algorithm in~\cite{GroheKreutzerSiebertz_2017_NowhereDense}
first computes a linear order with the following properties.

\begin{lemma}[{\cite[Corollary~6.8]{GroheKreutzerSiebertz_2017_NowhereDense}}]
  \label{lem:neighbourhood-cover-order}
  Let \(\C\) be an (effectively) nowhere dense class of graphs.
  There is a (computable) function \(f\) such that
  for all \(r \in \N\) and \(\delta > 0\)
  and every \(G \in \C\) with \(n \geq f(r, \delta)\) vertices,
  we can compute an order \(<\) on \(V(G)\)
  such that \(\bigabs{\WReach_r[G, <, v]} \leq n^\delta\)
  for all \(v \in V(G)\) in time \(f(r,\delta) \cdot n^{1+\delta}\).
\end{lemma}

The clusters of the neighbourhood cover computed
in~\cite[Theorem~6.2]{GroheKreutzerSiebertz_2017_NowhereDense}
are of the form \(X_{2r}[G, <, t] \deff \setc{s \in V(G)}{t \in \WReach_{2r}[G, <, s]}\)
with centre \(t\) for some vertex \(t \in V(G)\).
More precisely, \cite{GroheKreutzerSiebertz_2017_NowhereDense}
uses the following neighbourhood cover
(stated with the original notation from
\cite{GroheKreutzerSiebertz_2017_NowhereDense} for neighbourhood covers,
but replacing \enquote{maximum degree} with \enquote{overlap}).

\begin{lemma}[{\cite[Lemma~6.10]{GroheKreutzerSiebertz_2017_NowhereDense}}]
  \label{lem:neighbourhood-cover-construction}
  Let \(G\) be a graph, let \(<\) be a linear order on \(V(G)\),
  and let \(p, r \in \N\) such that
  \(\bigabs{\WReach_{2r}[G, <, v]} \leq p\) for all \(v \in V(G)\).
  For \(v \in V(G)\), let \(m(v)\) be the minimum of \(\neighbr{G}{v}\)
  with respect to \(<\).
  Then \(\vertexToCluster \colon V(G) \to 2^{V(G)}\)
  with \(\vertexToCluster(v) = X_{2r}[G, <, m(v)]\)
  is a distance-\(r\) neighbourhood cover of \(G\)
  with radius at most \(2r\) and overlap at most~\(p\).
\end{lemma}

Translated to the notation introduced in this paper,
the neighbourhood cover described in \cref{lem:neighbourhood-cover-construction}
is the neighbourhood cover \(\NCtuple{}\)
with centres \(\centres = \setc{m(v)}{v \in V(G)}\),
\(\vertexToCentre(v) = m(v)\) for all \(v \in V(G)\),
and \(\centreToCluster(\centre) = X_{2r}[G, <, \centre]\)
for all \(\centre \in \centres\).
The proof of \cite[Lemma~6.10]{GroheKreutzerSiebertz_2017_NowhereDense}
shows that \(\neighbr{G}{v} \subseteq \vertexToCluster(v)\)
for all \(v \in V(G)\),
so \(\NC\) is indeed a distance-\(r\) neighbourhood cover.
Based on the linear order
computed via \cref{lem:neighbourhood-cover-order}
applied with radius \(2r\) and \(\delta \deff \frac{\epsilon}{2}\),
this neighbourhood cover has overlap at most
\(\abs{V(G)}^{\epsilon/2} \leq \abs{V(G)}^{\epsilon}\),
which matches the requirement in \cref{thm:neighbourhood-cover}.
\(\NC\) having radius at most \(2r\) requires that
\(\centreToCluster(\centre) \subseteq
\neighb{2r}{G[\centreToCluster(\centre)]}{\centre}\),
which also implies that the clusters are connected.
From the construction, it is straightforward to see that the weaker statement
\(\centreToCluster(\centre) \subseteq \neighb{2r}{G}{\centre}\)
holds for every \(\centre \in \centres\).
Since the proof in~\cite[Lemma~6.10]{GroheKreutzerSiebertz_2017_NowhereDense}
does not explicitly show how to obtain the stronger statement,
we prove the following lemma that implies the stronger statement.

\begin{lemma}
  Let \(G\) be a graph, let \(<\) be a linear order on \(V(G)\),
  let \(k \in \N\), \(t \in V(G)\), and
  \(X \deff X_k[G, <, t] \deff \setc{s \in V(G)}{t \in \WReach_k[G, <, s]}\).
  For all \(s \in X\), it holds that \(\dist^{G[X]}(s, t) \leq k\).
\end{lemma}
\begin{proof}
  Let \(s \in X\).
  By definition, we have \(t \in \WReach_k[G, <, s] = \WReach_k(G, <, s) \cup \set{s}\).
  If \(t = s\), then \(\dist^{G[X]}(s, t) = 0 \leq k\).
  If \(t \neq s\), then \(t \in \WReach_k(G, <, s)\),
  so it holds that \(t < s\) and there is a path \(P = s \to u_1 \to \dots \to u_\ell \to t\)
  of length \(\ell + 1 \leq k\) in \(G\)
  such that \(t \leq u\) for all \(u \in V(P)\).
  Since \(<\) is a linear order and \(s, u_1, \dots, u_\ell, t\) are pairwise distinct,
  this implies that \(t < u_i\) for all \(i \in [\ell]\).
  Moreover, for all \(i \in [\ell]\), \(P_i \deff u_i \to u_{i+1} \to \dots \to u_\ell \to t\)
  is a path of length at most \(k\) in \(G\)
  with \(t \leq u\) for all \(u \in V(P_i)\),
  so \(t \in \WReach_k[G, <, u_i]\), and thus \(u_i \in X\).
  This shows that \(P \subseteq G[X]\), which implies \(\dist^{G[X]}(s, t) \leq \ell + 1 \leq k\).
\end{proof}

It remains to discuss the data structure
described in \cref{thm:neighbourhood-cover}.
It consists of a list of all vertices of \(G\),
a list of all centres in \(\centres\),
and a list of all clusters \(\centreToCluster(\centre)\)
for \(\centre \in \centres\).
Furthermore, we have links from \(v\) in the list of vertices
to \(\vertexToCentre(v)\) in the list of centres for all \(v \in V(G)\)
and links from \(\centre\) in the list of centres
to \(\centreToCluster(\centre)\) in the list of clusters
for every \(\centre \in \centres\).

For computing the list of centres and the links from vertices to centres,
we proceed as follows.
First, for every \(v \in V(G)\),
we set \(m_0(v) \deff v\).
Then, for \(i\) from \(1\) to \(r\)
and for all \(v \in V(G)\),
we set \(m_i(v) \deff \min \setc{m_{i-1}(w)}{w \in \neighb{1}{G}{v}}\)
(recall that \(v \in \neighb{1}{G}{v}\)),
where vertices are compared using the linear order \(<\)
computed via \cref{lem:neighbourhood-cover-order}.
In the end, we set \(\centres \deff \setc{m_r(v)}{v \in V(G)}\)
and \(\vertexToCentre(v) \deff m(v) \deff m_r(v)\) for all \(v \in V(G)\).
By induction, it can easily be verified that
\(m_i(v)\) is the minimum of \(\neighb{i}{G}{v}\) with respect to \(<\),
so \(\centres\) and \(\vertexToCentre\) indeed match the construction
in \cref{lem:neighbourhood-cover-construction}.
By \cref{lem:nowhere-dense-cliques},
for an effectively nowhere dense class \(\C\),
there is a computable function \(f'\) such that
\(\abs{V(G)} + \abs{E(G)} \leq f'(\epsilon) \cdot \abs{V(G)}^{1+\epsilon}\)
for all \(G \in \C\) and all \(\epsilon \in \Qpos\).
Hence, the computation of \(\centres\) and \(\vertexToCentre\)
takes time at most
\(\bigO\bigl(r \cdot (\abs{V(G)} + \abs{E(G)})\bigr)
\leq \bigO\bigl(r \cdot f'(\epsilon) \cdot \abs{V(G)}^{1+\epsilon}\bigr)\).

The proof of~\cite[Theorem~6.2]{GroheKreutzerSiebertz_2017_NowhereDense}
describes how to compute the sets \(X_{2r}[G, <, v]\) for all \(v \in V(G)\)
(and, in particular, for all \(\centre \in \centres\))
in total time \(f''(r, \epsilon) \cdot \abs{V(G)}^{1+\epsilon}\),
where \(f''\) is a computable function
if \(\C\) is effectively nowhere dense.
This provides the list of all clusters
as well as the mapping~\(\centreToCluster\).

The data structure described so far already handles the computation of
\(\centreToCluster(\centre)\) given a centre \(\centre\)
in time \(\bigO(\abs{\centreToCluster(\centre)})\)
and the computation of \(\vertexToCentre(v)\) given a vertex \(v \in V(G)\)
in time \(\bigO(1)\).
It remains to build a data structure for checking
in fixed-parameter-constant time
whether \(v \in \centreToCluster(\centre)\)
for a given \(v \in V(G)\) and \(\centre \in \centres\).
For this, we use the following result
from~\cite{SchweikardtSegoufinVigny_2022_Enumeration}.

\begin{theorem}[{Storing Theorem,
  \cite[Theorem~3.1]{SchweikardtSegoufinVigny_2022_Enumeration}}]
  \label{thm:storing}
  There is a computable function \(\fStore\)
  such that, for every \(k, n \in \N\) and \(\epsilon > 0\),
  there is a data structure that stores the value of a \(k\)-ary function \(\mu\)
  of domain \(\dom(\mu) \subseteq [n]^k\) with
  \begin{itemize}
    \item initialisation time
      \(\fStore(k, \epsilon) \cdot \abs{\dom(\mu)} \cdot n^\epsilon\),
    \item update time \(\fStore(k, \epsilon) \cdot n^\epsilon\)
      whenever an entry \((\ta, b)\) is added to or removed from \(\mu\),
    \item lookup time \(\fStore(k, \epsilon)\),
    \item and at any point in time, the space used by the data structure
      is \(\fStore(k, \epsilon) \cdot \abs{\dom(\mu)} \cdot n^{\epsilon}\).
  \end{itemize}
  Here, lookup means that, given a tuple \(\ta \in [n]^k\),
  the algorithm either answers \(b\) if \(\ta \in \dom(\mu)\) and \(\mu(\ta) = b\),
  or it answers \(\ta'\) if \(\ta \not\in \dom(\mu)\)
  and \(\ta' = \min \setc{\tx \in \dom(\mu)}{\tx > \ta}\),
  or it answers \emph{null} if no such tuple exists.
\end{theorem}
In~\cite{SchweikardtSegoufinVigny_2022_Enumeration},
this result is stated with a constant \(c \in \N\) for every fixed \(k\) and \(\epsilon\)
instead of a computable value \(\fStore(k, \epsilon)\).
However, it can be verified from the proof of the theorem
that the constant is indeed computable from \(k\) and \(\epsilon\).

For the remaining data structure,
for every centre \(\centre \in \centres\),
we use \cref{thm:storing} for storing the unary function \(\mu_\centre\)
with domain \(\dom(\mu_\centre) = \centreToCluster(\centre)\)
and values \(\mu_\centre(v) = 1\) for all \(v \in \centreToCluster(\centre)\).
Here, we assume \(V(G) = [\abs{V(G)}]\),
and we maintain a bijection between \(V(G)\) and \([\abs{V(G)}]\)
if this is not the case.
For a single cluster,
the initialisation of the data structure takes time
\(\fStore(1,\frac{\epsilon}{2}) \cdot \abs{\centreToCluster(\centre)}
\cdot \abs{V(G)}^{\epsilon/2}\),
and we can check in lookup time \(\fStore(1,\frac{\epsilon}{2})\)
whether a given vertex \(v \in V(G)\)
is contained in \(\centreToCluster(\centre)\).
As described above, the neighbourhood cover has overlap
at most \(\abs{V(G)}^{\epsilon/2}\),
so it holds that
\(\sum_{\centre \in \centres} \abs{\centreToCluster(\centre)}
\leq \abs{V(G)}^{1+\epsilon/2}\).
Combined with the initialisation time for a single cluster,
this shows that the computation of all data structures for checking
whether \(v \in \centreToCluster(\centre)\)
for any \(v \in V(G)\) and \(\centre \in \centres\)
takes time at most
\(\fStore(1,\frac{\epsilon}{2}) \cdot \abs{V(G)}^{1+\epsilon}\).

All in all, we can define a function \(\fNC\)
that satisfies \cref{thm:neighbourhood-cover};
if \(\C\) is effectively nowhere dense, then \(\fNC\) is computable.
 
\end{document}